\documentclass[a4paper,11pt]{amsart}
\usepackage[active]{srcltx}

\usepackage{graphicx}
\usepackage{amsmath}
\usepackage{amssymb}
\usepackage{hyperref}
\usepackage{bbm}

\usepackage{longtable}

\newtheorem{thm}{Theorem}[section]
\newtheorem{lem}[thm]{Lemma}%
\newtheorem{prop}[thm]{Proposition}%
\theoremstyle{remark}
\theoremstyle{plain}

\numberwithin{equation}{section}

\def\NN{{\mathbb N}}
\def\QQ{{\mathbb Q}}

\def\RR{{\mathbb R}}

\def\ZZ{{\mathbb Z}}

\def\vecb{{\text{\boldmath$b$}}}

\def\vece{{\text{\boldmath$e$}}}

\def\vecq{{\text{\boldmath$q$}}}
\def\vecQ{{\text{\boldmath$Q$}}}
\def\vecR{{\text{\boldmath$R$}}}

\def\vecs{{\text{\boldmath$s$}}}
\def\vecS{{\text{\boldmath$S$}}}

\def\vecu{{\text{\boldmath$u$}}}
\def\vecU{{\text{\boldmath$U$}}}
\def\vecv{{\text{\boldmath$v$}}}
\def\vecV{{\text{\boldmath$V$}}}
\def\vecw{{\text{\boldmath$w$}}}
\def\vecW{{\text{\boldmath$W$}}}
\def\vecx{{\text{\boldmath$x$}}}
\def\vecX{{\text{\boldmath$X$}}}

\def\vecY{{\text{\boldmath$Y$}}}
\def\vecz{{\text{\boldmath$z$}}}
\def\vecZ{{\text{\boldmath$Z$}}}
\def\vecalf{{\text{\boldmath$\alpha$}}}

\def\veceta{{\text{\boldmath$\eta$}}}

\def\vecxi{{\text{\boldmath$\xi$}}}

\def\uvecalf{{\text{\underline{\boldmath$\alpha$}}}}
\def\uveceta{{\text{\underline{\boldmath$\eta$}}}}

\def\vecnull{{\text{\boldmath$0$}}}

\def\scrA{{\mathcal A}}
\def\scrB{{\mathcal B}}

\def\scrH{{\mathcal H}}
\def\scrI{{\mathcal I}}

\def\scrK{{\mathcal K}}
\def\scrL{{\mathcal L}}
\def\scrM{{\mathcal M}}
\def\scrN{{\mathcal N}}

\def\scrP{{\mathcal P}}

\def\scrS{{\mathcal S}}

\def\scrV{{\mathcal V}}

\def\txi{\tilde{\xi}}

\def\e{\mathrm{e}}

\def\id{\operatorname{id}}

\def\C{\operatorname{C{}}}

\def\L{\operatorname{L{}}}

\def\S{\operatorname{S{}}}

\def\SO{\operatorname{SO}}

\def\O{\operatorname{O{}}}
\def\T{\operatorname{T{}}}
\def\tr{\operatorname{tr}}

\def\meas{\operatorname{meas}}
\def\Var{\operatorname{Var}}

\def\Ran{\operatorname{Ran}}

\def\trans{\,^\mathrm{t}\!}

\def\xibar{\overline{\xi}}
\def\sigmabar{\overline{\sigma}}

\def\hatw{{\widehat{\vecw}}}

\def\UB{{\scrB_1^{d-1}}}
\def\US{{\S_1^{d-1}}}

\def\Prob{\operatorname{\mathbf{P}}}
\def\Exp{\operatorname{\mathbf{E}}}
\def\Var{\operatorname{\mathbf{Var}}}
\def\Cov{\operatorname{\mathbf{Cov}}}

\newcommand{\prob}[1]{\ensuremath{\mathbf{P}\big(#1\big)}}
\newcommand{\expect}[1]{\ensuremath{\mathbf{E}\big(#1\big)}}
\newcommand{\var}[1]{\ensuremath{\mathbf{Var}\big(#1\big)}}
\newcommand{\cov}[2]{\ensuremath{\mathbf{Cov}\big(#1,#2\big)}}

\newcommand{\condexpect}[2]{\ensuremath{\mathbf{E}\big(#1\bigm|#2\big)}}

\newcommand{\ind}[1]{\ensuremath{{\mathbbm{1}}_{\{#1\}}}}
\def \toprob {\,\,\buildrel\Prob\over\longrightarrow\,\,}
\def \toas {\,\,\buildrel\text{a.s.}\over\longrightarrow\,\,}

\def\vareps{\varepsilon}
\title[Superdiffusion in the periodic Lorentz gas]{Superdiffusion in the periodic Lorentz gas}
\author{Jens Marklof}
\author{B\'alint T\'oth}
\address{Jens Marklof, School of Mathematics, University of Bristol,
Bristol BS8 1TW, U.K.\newline
\rule[0ex]{0ex}{0ex} \hspace{8pt}{\tt j.marklof@bristol.ac.uk}}
\address{B\'alint T\'oth, School of Mathematics, University of Bristol,
Bristol BS8 1TW, U.K.; MTA-BME Stochastics Research Group, Budapest, Hungary; R\'enyi Institute, Budapest, Hungary\newline
\rule[0ex]{0ex}{0ex} \hspace{8pt}{\tt balint.toth@bristol.ac.uk,  balint@math.bme.hu}}

\date{24 March 2014; revised and expanded 16 November 2015}
\thanks{The research leading to these results has received funding from the European Research Council under the European Union's Seventh Framework Programme (FP/2007-2013) / ERC Grant Agreement n. 291147. 
J.M.\ is furthermore supported by a Royal Society Wolfson Research Merit Award. The research of B.T. is partially supported by the Hungarian National Science Foundation (OTKA) through grant K100473 and by the Leverhulme Trust through International Network Grant ``Laplacians, Random Walks, Quantum Spin Systems.'' Both authors thank the Isaac Newton Institute, Cambridge for its support and hospitality during the programmes ``Periodic and Ergodic Spectral Problems'' and ``Random Geometry.''}


\begin{document}

\begin{abstract}
We prove a superdiffusive central limit theorem for the displacement of a test particle in the periodic Lorentz gas in the limit of large times $t$ and low scatterer densities (Boltzmann-Grad limit). The normalization factor is $\sqrt{t\log t}$, where $t$ is measured in units of the mean collision time. This result holds in any dimension and for a general class of finite-range scattering potentials. We also establish the corresponding invariance principle, i.e., the weak convergence of the particle dynamics to Brownian motion.
\end{abstract}

\maketitle

\section{Introduction}

The periodic Lorentz gas is one of the iconic models of ``chaotic'' diffusion in deterministic systems. It describes the dynamics of a test-particle in an infinite periodic array of spherically symmetric scatterers. The main results characterizing the diffusive nature of the periodic Lorentz gas have to date been mainly restricted to the two-dimensional setting and hard-sphere scatterers. The first seminal result on this subject was the proof of a central limit theorem for the displacement of the test particle at large times $t$ for the finite-horizon Lorentz gas by Bunimovich and Sinai \cite{Bunimovich:1980ur}. For more general invariance principles see Melbourne and Nicol \cite{Melbourne:2009ju} and references therein. In the case of the infinite-horizon Lorentz gas, Bleher \cite{Bleher:1992ku} pointed out that the mean-square displacement grows like $t\log t$ when $t\to\infty$, as opposed to a linear growth in the finite-horizon case. The superdiffusive central limit theorem suggested in \cite{Bleher:1992ku} was first proved by Sz\'asz and Varj\'u \cite{Szasz:2007uo} for the discrete-time billiard map. Dolgopyat and Chernov \cite{Dolgopyat:2009bl} provided an alternative proof, and established the central limit theorem and invariance principle for the billiard flow. Analogous results hold for the stadium billiard (B\'alint and Gou\"ezel \cite{Balint:2006jv}) and billiards with cusps (B\'alint, Chernov and Dolgopyat \cite{Balint:2011jt}). The difficulty in extending the above findings to dimensions greater than two lies in the possibly exponential growth of the complexity of singularities (B\'alint and T\'oth \cite{Balint:2008kt,Balint:2012fg}, Chernov \cite{Chernov:1994hb}) and, in the case of infinite horizon, the subtle geometry of channels (Dettmann \cite{Dettmann:2012bj}, N\'andori, Sz\'asz and Varj\'u \cite{2012arXiv1210.2231N}).

In the present paper we prove unconditional superdiffusive central limit theorems and invariance principles for the periodic Lorentz gas in any dimension $d\geq 2$, valid in the limit of low scatterer density (Boltzmann-Grad limit) and for a general class of finite-range scattering potentials. That is, instead of fixing the radius $r$ of each scatterer and considering the long time limit as in the above cited papers, we consider here the limit $r\to 0$ and then the limit of long times, where time is measured in units of the mean collision time. It is an interesting open problem to consider the two limits $r\to 0$, $t\to\infty$ jointly.

The precise setting of our study is as follows. Let $\scrL\subset\RR^d$ be a fixed Euclidean lattice of covolume one (such as the cubic lattice $\scrL=\ZZ^d$), and define the scaled lattice $\scrL_r:=r^{(d-1)/d}\scrL$. At each point in $\scrL_r$ we center a sphere of radius $r$. We consider a test particle that moves along straight lines with unit speed until it hits a sphere, where it is scattered elastically. The above scaling of scattering radius vs.\ lattice spacing ensures that the mean free path length (i.e., the average distance between consecutive collisions) has the limit $\xibar=1/v_{d-1}$ as $r\to 0$, where $v_{d-1}=\pi^{\frac{d-1}{2}}/\Gamma(\frac{d+1}{2})$ denotes the volume of the unit ball in $\RR^{d-1}$. 

In the case of the classic Lorentz gas the scattering mechanism is given by specular reflection, but as in \cite{Marklof:2011ho} we will here also allow more general spherically symmetric scattering maps. The precise conditions will be stated in Section \ref{sec:two_a}. 

The position of our test particle at time $t$ is denoted by
\begin{equation}
\vecx_t=\vecx_t(\vecx_0,\vecv_0) \in \scrK_r := \RR^d \setminus (\scrL_r+r\scrB_1^d) ,
\end{equation}
where $\vecx_0$ and $\vecv_0$ are position and velocity at time $t=0$, and $\scrB_1^d$ is the open unit ball in $\RR^d$ centered at the origin.
We use the convention that for any boundary point $\vecx_0 \in \partial \scrK_r$ we choose the {\em outgoing} velocity $\vecv_0$, i.e.\ the velocity {\em after} the scattering. The corresponding phase space is denoted by $\T^1(\scrK_r)$. For notational reasons it is convenient to extend the dynamics to $\T^1(\RR^d):=\RR^d\times\US$ by setting $\vecx_t=\vecx_0$ for all initial conditions $\vecx_0\notin\scrK_r$.

We consider the time evolution of a test particle with random initial data $(\vecx_0,\vecv_0)\in\T^1(\RR^d)$, distributed according to a given Borel probability measure $\Lambda$ on $\T^1(\RR^d)$. 
The following superdiffusive central limit theorem, valid for small scattering radii and large times, asserts that the normalized particle displacement at time $t$, and measured in units of $\sqrt{t\log t}$, converges weakly to a Gaussian distribution.

\begin{thm}\label{thm:main1}
Let $d\geq 2$ and fix a Euclidean lattice $\scrL\subset\RR^d$ of covolume one. Assume $(\vecx_0,\vecv_0)$ is distributed according to an absolutely continuous Borel probability measure $\Lambda$ on $\T^1(\RR^d)$. 
Then, taking first $r\to 0$ and then $t\to\infty$, we have
\begin{equation}\label{eq:main1}
\frac{\vecx_t -\vecx_0}{\Sigma_d \sqrt{t\log t}}  \Rightarrow \scrN(0,I_d) ,
\end{equation}
where $\scrN(0,I_d)$ is a centered normal random variable in $\RR^d$ with identity covariance matrix, and
\begin{equation}
\Sigma_d^2:=\frac{2^{1-d} v_{d-1}}{d^2(d+1)\zeta(d)}.
\end{equation}
\end{thm}
Here $\zeta(d) := \sum_{n=1}^\infty n^{-d}$ denotes the Riemann zeta function. 
Recall that the weak convergence \eqref{eq:main1} holds if and only if
\begin{equation}\label{eq:main1explicit}
\lim_{t\to\infty}\lim_{r\to 0} \Exp f\bigg(\frac{\vecx_t -\vecx_0}{\Sigma_d \sqrt{t\log t}}  \bigg) =
\frac{1}{(2\pi)^{d/2}} \int_{\RR^d} f(\vecx)\, \e^{-\frac12 \|\vecx\|^2} d\vecx 
\end{equation}
for any bounded continuous $f:\RR^d\to\RR$.

Theorem \ref{thm:main1} will follow from its descrete-time analogue, Theorem \ref{thm:main2}. Let us denote by $\vecq_n=\vecq_n(\vecq_0,\vecv_0)\in\partial\scrK_r$ ($n=1,2,3,\ldots$) the location where the test particle with initial condition $(\vecq_0,\vecv_0)$ leaves the $n$th scatterer. It is natural in this setting to assume $\vecq_0\in\partial\scrK_r$. By the translational invariance of the lattice, we may in fact assume without loss of generality $\vecq_0\in r \US$. 
For given exit velocity $\vecv_0$, we write 
\begin{equation}
\vecq_0 = r (\vecs_0+\vecv_0 \sqrt{1-\|\vecs_0\|^2}) 
\end{equation}
and stipulate in the following that the random variable $\vecs_0$ is uniformly distributed in the unit disc orthogonal to $\vecv_0$. The uniform distribution is the natural invariant measure for the discrete time dynamics.

\begin{thm}\label{thm:main2}
Let $d\geq 2$ and $\scrL$ as above. Assume $\vecv_0$ is distributed according to an absolutely continuous Borel probability measure $\lambda$ on $\US$. Then, taking first $r\to 0$ and then $n\to\infty$, we have
\begin{equation}\label{eq:main2}
\frac{\vecq_n -\vecq_0}{\sigma_d \sqrt{n\log n}} 
\Rightarrow \scrN(0,I_d)  ,
\end{equation}
with
\begin{equation}
\sigma_d^2:=\frac{2^{1-d}}{d^2(d+1)\zeta(d)} =\xibar\, \Sigma_d^2.
\end{equation}
\end{thm}

The above results generalise to functional central limit theorems, also known as invariance principles. 
Denote by $\C_0([0,1])$ the space of curves $[0,1]\to\RR^d$ starting at the origin. We fix a metric on $\C_0([0,1])$ by defining the distance between two curves $\vecX_1$ and $\vecX_2$ by $\sup_{t\in[0,1]} \|\vecX_1(t)-\vecX_2(t)\|$. The topology generated by open balls in this metric is called the {\em uniform topology.} A sequence $(\vecX_n)_n$ of random curves in $\C_0([0,1])$ {\em converges weakly} to $\vecX$ ($\vecX_n\Rightarrow \vecX$), if for any bounded continuous $f:\C_0([0,1])\to\RR$ we have $\lim_n \Exp f(\vecX_n)=\Exp f(\vecX)$.

The following theorem, which is the main result of this paper, states that for the same random initial data as in Theorem \ref{thm:main1}, the random curves 
\begin{equation}
[0,1]\to\RR^d, \qquad t\mapsto \vecX_{T,r}(t) := \frac{\vecx_{t T} -\vecx_0}{\Sigma_d \sqrt{T\log T}},
\end{equation}
converge weakly to the standard Brownian motion $t\mapsto \vecW(t)$ in $\RR^d$ with unit covariance matrix $I_d$. 

\begin{thm}\label{plusthm:main1}
Let $d\geq 2$ and fix a Euclidean lattice $\scrL\subset\RR^d$ of covolume one. Assume $(\vecx_0,\vecv_0)$ is distributed according to an absolutely continuous Borel probability measure $\Lambda$ on $\T^1(\RR^d)$. Then, taking first $r\to 0$ and then $T\to\infty$, we have
\begin{equation}\label{pluseq:main1}
 \vecX_{T,r}  \Rightarrow \vecW .
\end{equation}
\end{thm}

As in the case of Theorem \ref{thm:main1}, we derive Theorem \ref{plusthm:main1} as a corollary of its discrete-time analogue, Theorem \ref{plusthm:main2}. By linearly interpolating between the position variables $\vecq_0,\vecq_1,\ldots,\vecq_n$, we obtain the piecewise linear curve
\begin{equation}\label{pewi}
[0,1]\to\RR^d, \qquad t\mapsto \vecq_n(t):= \vecq_{\lfloor nt\rfloor} + \{ nt\}\, \big( \vecq_{\lfloor nt\rfloor+1} - \vecq_{\lfloor nt\rfloor} \big)  ,
\end{equation}
where  $\{ x \}:= x-\lfloor x \rfloor$ denotes the fractional part of $x$. We rescale the curve by setting
\begin{equation}
\vecY_{n,r}(t):=\frac{\vecq_n(t)-\vecq_0}{\sigma_d \sqrt{n\log n}} .
\end{equation}
%The following theorem is then the discrete-time analogue of Theorem \ref{plusthm:main2}.
We then have the following generalization of Theorem \ref{thm:main2}.

\begin{thm}\label{plusthm:main2}
Let $d\geq 2$ and $\scrL$ a Euclidean lattice of covolume one. Assume $\vecv_0$ is distributed according to an absolutely continuous Borel probability measure $\lambda$ on $\US$. Then, taking first $r\to 0$ and then $n\to\infty$, we have
\begin{equation}\label{pluseq:main2}
\vecY_{n,r}\Rightarrow \vecW .
\end{equation}
\end{thm}

The starting point of our analysis is the paper \cite{Marklof:2011ho}, which proves that, for every fixed $t>0$, the limit $r\to 0$ in \eqref{eq:main1} (resp.~\eqref{eq:main2}) exists and is given by a continuous-time (resp.~discrete-time) Markov process. The main objective of the present study is therefore to prove a superdiffusive central limit theorem, as well as an invariance principle, for each of these Markov processes. The central limit theorem is stated as Theorem \ref{thm:main3} in Section \ref{sec:two_b} after a brief survey of the relevant results from \cite{Marklof:2011ho}. The subsequent sections of the paper are devoted to the proof of Theorem \ref{thm:main3}. The invariance principles stated in Theorems \ref{plusthm:main1} and \ref{plusthm:main2} follow from the results in Sections \ref{plussec:finite}--\ref{plussec:221}.

\section{The scattering map}\label{sec:two_a}

\begin{figure}
\begin{center}
\begin{minipage}{0.8\textwidth}
\unitlength0.1\textwidth
\begin{picture}(10,5)(0,0)
\put(0,0){\includegraphics[width=\textwidth]{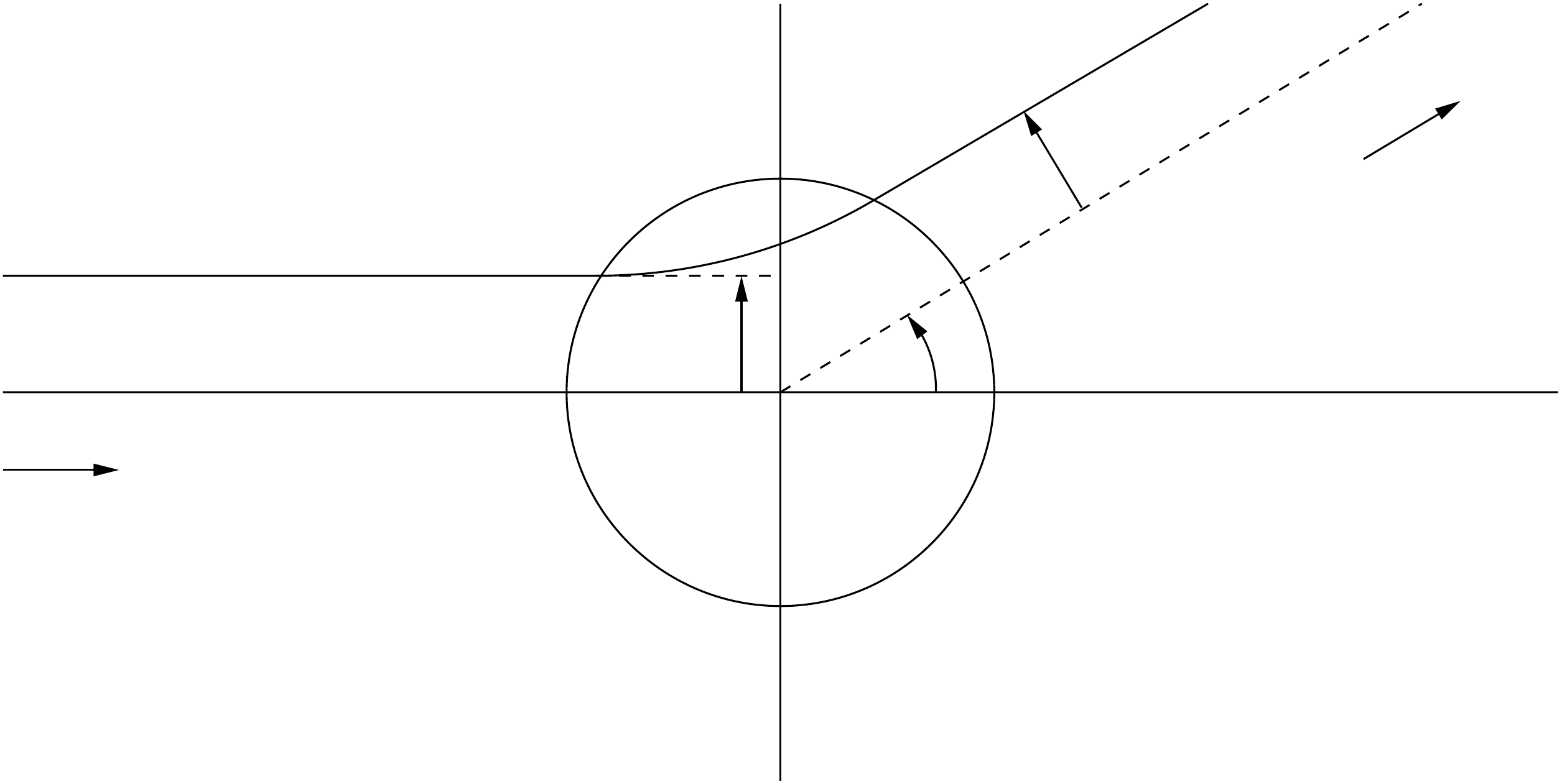}}
\put(0,1.75){$\vecv_{-}$} 
\put(4.4,2.75){$\vecb$}
\put(5.6,2.6){$\theta$} 
\put(6.8,4.1){$\vecs$}
\put(9.25,4){$\vecv_+$}
\end{picture}
\end{minipage}
\end{center}
\caption{The scattering map.} \label{fig1}
\end{figure}
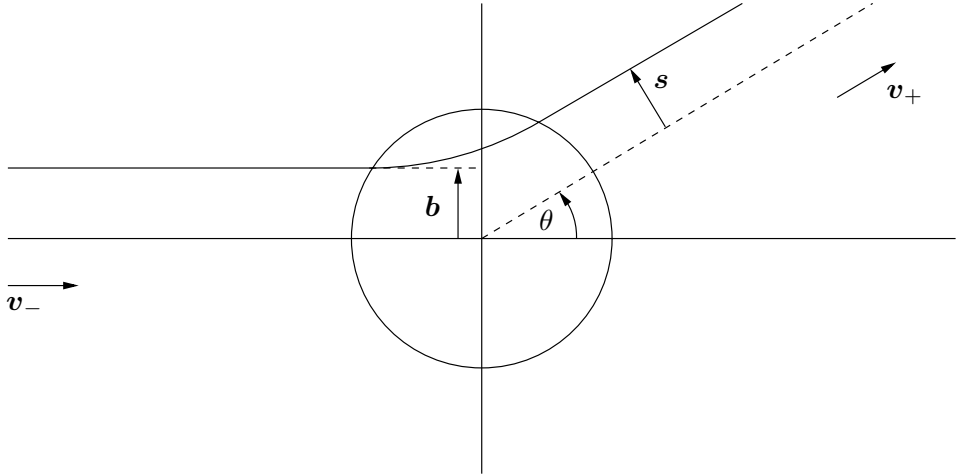

We now specify the conditions on the scattering map that are assumed in Theorems \ref{thm:main1}--\ref{plusthm:main2}. These are the same as in \cite{Marklof:2011ho}, with the additional simplifying assumption that the scattering map preserves angular momentum, cf.~\cite[Remark 2.3]{Marklof:2011ho}. We describe the scattering map in units of $r$, i.e., the scatterer is represented as the open unit ball $\scrB_1^d$. Set
\begin{equation}
\scrS:=\{ (\vecv,\vecb)\in\US\times \scrB_1^d \mid \vecv\cdot \vecb=0\} ,
\end{equation}
and consider the scattering map
\begin{equation}\label{scatmap}
\Theta: \scrS\to\scrS,\qquad (\vecv_-,\vecb)\mapsto (\vecv_+,\vecs).
\end{equation}
The incoming data is denoted by $(\vecv_-,\vecb)\in\scrS$, where $\vecv_-$ is the velocity of the particle before the collision and $\vecb$ the impact parameter, i.e., the point of impact on $\US$ projected onto the plane $\{\vecb\in\RR^d \mid \vecv_-\cdot \vecb=0\}$.
The outgoing data is analogously defined as $(\vecv_+,\vecs)\in\scrS$, where $\vecv_+$ is the velocity of the particle after the collision and $\vecs$ the exit parameter, cf.~Figure \ref{fig1}. Since we assume the scattering map is spherically symmetric, it is sufficent to define $\Theta$ for $(\vecv_-,\vecb)=(\vece_1,w\vece_2)$ for $w\in[0,1)$, where $\vece_j$ denotes the unit vector in the $j$th coordinate direction.
Any spherically symmetric scattering map \eqref{scatmap} which preserves angular momentum is thus uniquely determined by
\begin{equation}
\Theta(\vece_1,w\vece_2) = \big(\vece_1 \cos\theta(w) +  \vece_2 \sin\theta(w), -\vece_1 w\sin\theta(w) +  \vece_2 w\cos\theta(w) \big)
\end{equation}
where $\theta(w)$ is called the {\em scattering angle}. 

\begin{figure}
\begin{center}
\begin{minipage}{0.8\textwidth}
\unitlength0.1\textwidth
\begin{picture}(10,6.3)(0,0)
\put(0.3,0){\includegraphics[width=0.9\textwidth]{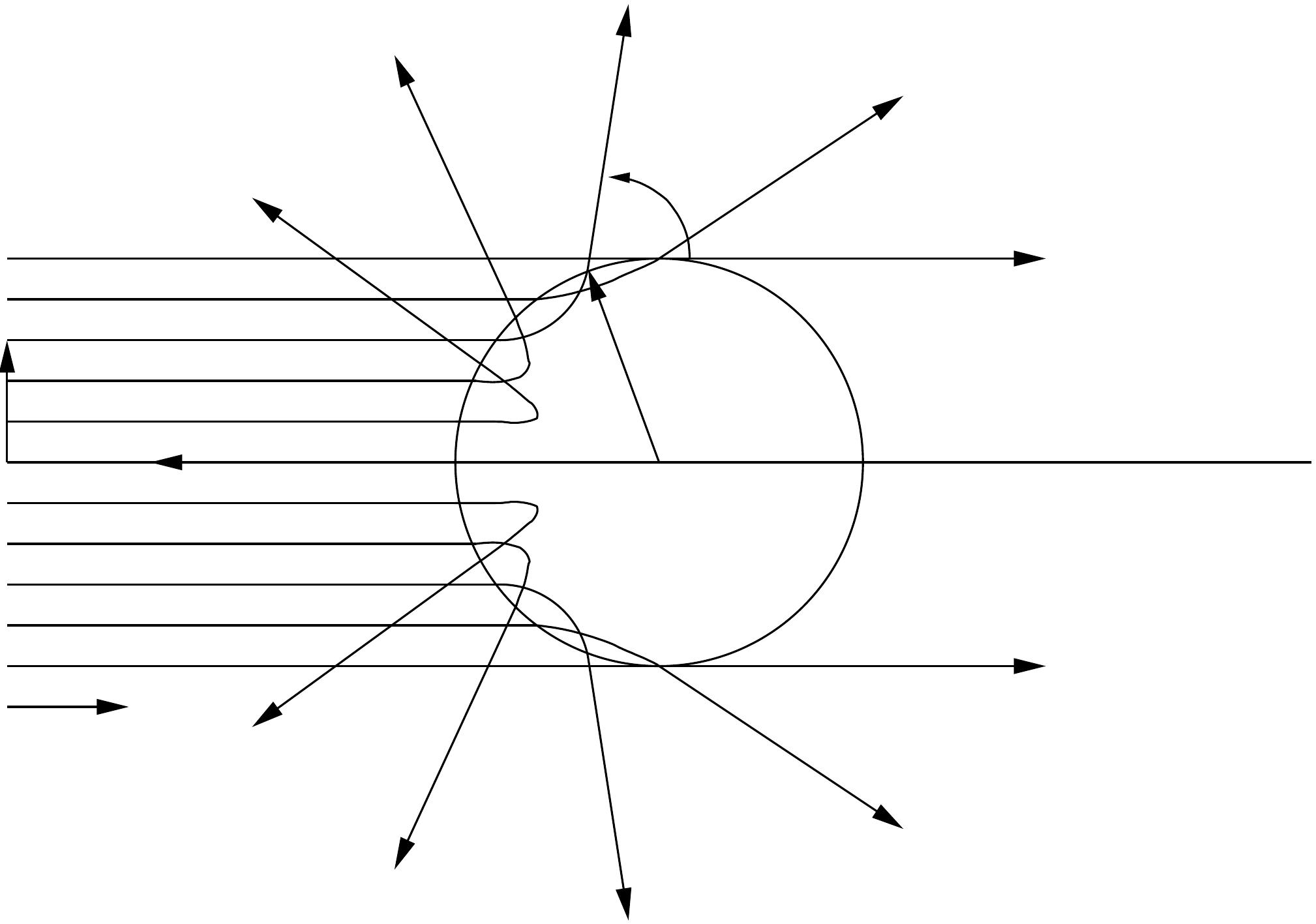}}
\put(0,3.5){$\vecb$} 
\put(0.5,1.2){$\vecv_-$} 
\put(4.6,4.65){$\theta$}
\put(4.6,5.8){$\vecv_+$} 
\end{picture}
\end{minipage}
\end{center}
\caption{Illustration of a scattering map satisfying Hypothesis (A).} \label{fig4}
\end{figure}
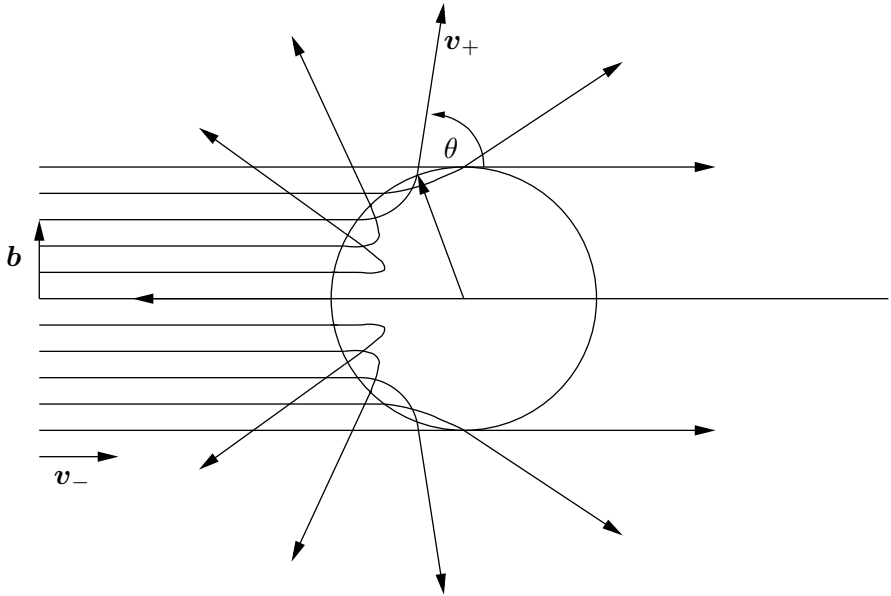

To satisfy the conditions of \cite{Marklof:2011ho}, we assume in the statements of Theorems  \ref{thm:main1} and \ref{thm:main2} that one of the following hypotheses is true (cf.~Fig.~\ref{fig4}):
\begin{itemize}
\item[(A)] {\em $\theta\in\C^1([0,1))$ is strictly decreasing with $\theta(0)=\pi$ and $\theta(w)>0$};
\item[(B)] {\em $\theta\in\C^1([0,1))$ is strictly increasing with $\theta(0)=-\pi$ and $\theta(w)<0$.}  
\end{itemize}

This assumption holds for a large class of scattering potentials, including muffin-tin Coulomb potentials, cf.~\cite{Marklof:2011ho}. In the case of hard-sphere scatterers we have $\theta(w)=\pi - 2 \arcsin(w)$ and hence Hypothesis (A) holds. 
For later use we define the minimal deflection angle by
\begin{equation}\label{Btheta}
B_\theta:= \inf_{w\in[0,1)} |\theta(w)|.
\end{equation}

Note that for more general impact parameters of the form
\begin{equation}
\vecb=\begin{pmatrix} 0 \\ \vecw  \end{pmatrix}, \qquad \vecw\in\UB\setminus\{\vecnull\}, 
\end{equation}
we have (by spherical symmetry)
\begin{equation}
\Theta\bigg(\begin{pmatrix} 1 \\ \vecnull  \end{pmatrix},\begin{pmatrix} 0 \\ \vecw  \end{pmatrix}\bigg)= \bigg(S(\vecw)\begin{pmatrix} 1 \\ \vecnull  \end{pmatrix}, S(\vecw)\begin{pmatrix} 0 \\ \vecw  \end{pmatrix} \bigg)
\end{equation}
with the matrix
\begin{equation}\label{SbLG}
S(\vecw) = E\big(\theta(w) \hatw\big), 
\end{equation}
where
\begin{equation}\label{SbLG222}
w:=\|\vecw\|>0, \qquad
\hatw:=w^{-1} \vecw \in\US,\qquad 
E(\vecx):= \exp\begin{pmatrix} 0 & -\trans\vecx \\ \vecx & 0_{d-1} \end{pmatrix}\in\SO(d) .
\end{equation}
More explicitly, 
\begin{equation}
S(\vecw) = \begin{pmatrix}
\cos\theta(w) & -\trans\hatw \sin\theta(w)  \\
\hatw \sin\theta(w) & 1_{d-1}-\hatw \otimes\hatw (1-\cos\theta(w))  
\end{pmatrix} .
\end{equation}
We extend the definition of $S(\vecw)$  to $\vecw=\vecnull$ by setting $S(\vecnull):=-I_d\in\SO(d)$ for $d$ even and 
$S(\vecnull):=\big(\begin{smallmatrix} -I_{d-1} & \\ & 1 \end{smallmatrix}\big) \in\SO(d)$ for $d$ odd. This choice ensures that $S(\vecnull)\vece_1=-\vece_1$.

For the case of general initial data $(\vecv_-,\vecb)\in\scrS$, assume $R(\vecv_-)\in\SO(d)$ and $\vecw\in\UB$ are chosen so that
\begin{equation}\label{TT}
\vecv_-=R(\vecv_-)\begin{pmatrix} 1 \\ \vecnull  \end{pmatrix} , \qquad \vecb=R(\vecv_-) \begin{pmatrix} 0 \\ \vecw  \end{pmatrix}.
\end{equation}
Then
\begin{equation}\label{TTT}
\Theta(\vecv_-,\vecb)= \bigg(R(\vecv_-) S(\vecw)\begin{pmatrix} 1 \\ \vecnull  \end{pmatrix}, R(\vecv_-) S(\vecw) \begin{pmatrix} 0 \\ \vecw  \end{pmatrix} \bigg) .
\end{equation}

We use an inductive argument to work out the velocity $\vecv_n$ after the $n$th collision, as well as the impact and exit parameters $\vecb_n$ and $\vecs_n$ of the $n$th collision.

\begin{lem}
Fix $\vecv_0$ and $R_0\in\SO(d)$ so that $\vecv_0=R_0\vece_1$, and denote by $(\vecv_n)_{n\in\NN}$, $(\vecb_n)_{n\in\NN}$, $(\vecs_n)_{n\in\NN}$ the sequence of velocities, impact and exit parameters of a given particle trajectory. Then there is a unique sequence $(\vecw_n)_{n\in\NN}$ in $\UB$ such that for all $n\in\NN$
\begin{equation}
\vecv_n=R_n \begin{pmatrix} 1 \\ \vecnull  \end{pmatrix} , \qquad \vecb_n=R_{n-1} \begin{pmatrix} 0 \\ \vecw_n  \end{pmatrix},
\qquad \vecs_n=R_n \begin{pmatrix} 0 \\ \vecw_n  \end{pmatrix},
\end{equation}
where
\begin{equation}
R_n := R_0 S(\vecw_1)\cdots S(\vecw_n).
\end{equation}
\end{lem}

\begin{proof}
We proceed by induction. We have $\vecv_0\cdot\vecb_1=0$ and thus $\vece_1\cdot R_0^{-1} \vecb_1=0$. We define $\vecw_1\in\UB$ by
\begin{equation}
\begin{pmatrix} 0 \\ \vecw_1  \end{pmatrix} = R_0^{-1} \vecb_1 .
\end{equation}
Then the assumption \eqref{TT} is satisfied and \eqref{TTT} yields
\begin{equation}\label{TTT2}
(\vecv_1,\vecs_1)=\Theta(\vecv_0,\vecb_1)= \bigg(R_0 S(\vecw_1)\begin{pmatrix} 1 \\ \vecnull  \end{pmatrix}, R_0 S(\vecw_1) \begin{pmatrix} 0 \\ \vecw_1  \end{pmatrix} \bigg) .
\end{equation}
which proves the case $n=1$. Let us therefore assume the statement is true for $n= k-1$. 
By the induction hypothesis, we have $\vecv_{k-1}=R_{k-1}\vece_1$. Note that $\vecv_{k-1}\cdot\vecb_k=0$ implies $\vece_1\cdot R_{k-1}^{-1}\vecb_k=0$, and define $\vecw_k\in\UB$ by
\begin{equation}
\begin{pmatrix} 0 \\ \vecw_k  \end{pmatrix} = R_{k-1}^{-1} \vecb_k .
\end{equation}
Therefore \eqref{TT} holds with $\vecv_-=\vecv_{k-1}$, $\vecb=\vecb_k$, and we can apply \eqref{TTT}:
\begin{equation}\label{TTT3}
\begin{split}
(\vecv_k,\vecs_k) & =\Theta(\vecv_{k-1},\vecb_{k})\\
& = \bigg(R_{k-1} S(\vecw_{k})\begin{pmatrix} 1 \\ \vecnull  \end{pmatrix}, R_{k-1} S(\vecw_{k}) \begin{pmatrix} 0 \\ \vecw_{k}  \end{pmatrix} \bigg) \\
& = \bigg(R_{k} \begin{pmatrix} 1 \\ \vecnull  \end{pmatrix}, R_{k} \begin{pmatrix} 0 \\ \vecw_{k}  \end{pmatrix} \bigg) ,
\end{split}
\end{equation}
where $R_k:=R_{k-1} S(\vecw_{k})=R_0 S(\vecw_1)\cdots S(\vecw_k)$. This completes the proof.
\end{proof}

\section{The Boltzmann-Grad limit}\label{sec:two_b}

We now recall the results of \cite{Marklof:2010ib,Marklof:2011ho} that are relevant to our investigation. Define the Markov chain
\begin{equation}\label{Markov}
n\mapsto ( \xi_n ,\veceta_n) 
\end{equation}
on the state space $\RR_{> 0} \times \UB$ with transition probability 
\begin{equation}\label{transprob}
\Prob\big( (\xi_n,\veceta_n)\in\scrA \bigm| \xi_{n-1},\veceta_{n-1} \big) 
= \int_\scrA \Psi_0(\veceta_{n-1}, x,\vecz) \, d x\,d\vecz .
\end{equation}
We will discuss the transition kernel $\Psi_0( \vecw,x,\vecz)$ in detail in Section \ref{sec:transition}. At this point, it sufficies to note that it is independent of $\xi_{n-1}$ and symmetric, i.e.~$\Psi_0( \vecw,x,\vecz)=\Psi_0(\vecz,x,\vecw)$. It is also independent of the choice of the scattering angle $\theta$, the lattice $\scrL$ and the initial particle distribution $\Lambda$ \cite{Marklof:2010ib}. (Note that $\Psi_0$ is related to the kernel $\Phi_0$ studied in \cite{Marklof:2010ib,Marklof:2011ho,Marklof:2011di} by $\Psi_0( \vecw,x,\vecz) = \Phi_0( x,\vecw,-\vecz)$.) Let
\begin{equation}
\Psi_0(x, \vecz):=\frac{1}{v_{d-1}} \int_\UB   \Psi_0( \vecw,x,\vecz) \, d\vecw ,
\end{equation}
\begin{equation}\label{rel2}
\Psi(x, \vecz):= \frac{1}{\,\xibar\,} \int_ x^\infty  \Psi_0(x',\vecz) \, d x' ,
\end{equation}
with the mean free path length $\xibar=1/v_{d-1}$.
Both $\Psi_0(x, \vecz)$ and $\Psi(x, \vecz)$ define probability densities on $\RR_{>0}\times\UB$ with respect to $d x\,d\vecz$. The first fact follows from the symmetry of the transition kernel, and the second from the relation
\begin{equation}
\int_{\UB\times\RR_{>0}} \Psi(x,\vecz) \, d x\,d\vecz = \frac{1}{\,\xibar\,} \int_{\UB\times\RR_{>0}}  x \Psi_0(x, \vecz) \, d x\,d\vecz
%= \frac{1}{\,\xibar\,} \int_{\UB\times\RR_{>0}}  x \Psi_0(x) \, d x  
= 1 .
\end{equation}
Suppose in the following that the sequence of random variables 
\begin{equation}\label{xieta}
\big( (\xi_n,\veceta_n) \big)_{n=1}^\infty
\end{equation}
is given by the Markov chain \eqref{Markov}, where $(\xi_1,\veceta_1)$ has density either $\Psi(x, \vecz)$ (for the continuous time setting) or $\Psi_0(x, \vecz)$ (for the discrete time setting). The relation \eqref{rel2} between the two reflects the fact that the continuous time Markov process is a suspension flow over the discrete time process, where the particle moves with unit speed between consecutive collisions; see \cite[Sect.~6]{Marklof:2011ho} for more details.  

We assume in the following that $R$ is a function $\US\to\SO(d)$ which satisfies $\vecv=R(\vecv)\vece_1$ and which is smooth when restricted to $\US\setminus\{-\vece_1\}$. An example is
\begin{equation}
R(\vecv)=
E\Bigl(\frac{2\arcsin\bigl(\|\vecv-\vece_1\|/2\bigr)}
{\|\vecv_\perp\|}\, \vecv_\perp\Bigr) \qquad \text{for} \quad
\vecv\in\US\setminus\{\vece_1,-\vece_1\},
\end{equation}
where $\vecv_\perp:=(v_2,\ldots,v_d)\in\RR^{d-1}$, and $R(\vece_1)=I$, $R(-\vece_{1})=-I$.

For $n\in\NN$, define the following random variables: 
\begin{equation}\label{set1}
\tau_n := \sum_{j=1}^n \xi_j,\quad \tau_0:=0, \qquad \text{(time to the $n$th collision);}
\end{equation}
\begin{equation}
\nu_t := \max\{ n\in\ZZ_{\geq 0} : \tau_n \leq t \} \qquad \text{(number of collisions within time $t$);}
\end{equation}
\begin{equation}\label{vac}
\vecV_n := R(\vecv_0) S(\veceta_1) \cdots S(\veceta_n) \vece_1,\quad \vecV_0:=\vecv_0, \qquad \text{(velocity after the $n$th collision);}
\end{equation}
\begin{equation}
\vecQ_n := \sum_{j=1}^n \xi_j \vecV_{j-1}  \qquad \text{(discrete time displacement);}
\end{equation}
\begin{equation}\label{set2}
\vecX_t :=  \vecQ_{\nu_t} + (t-\tau_{\nu_t})\vecV_{\nu_t} \qquad \text{(continuous time displacement).}
\end{equation}

\begin{thm}[\cite{Marklof:2011ho}]\label{thm:MS}
(i) Under the hypotheses of Theorem \ref{thm:main1}, for any $t>0$,
\begin{equation}\label{Xt}
\vecx_t -\vecx_0 \Rightarrow \vecX_t
\end{equation}
as $r\to 0$, where the random variable $(\xi_1,\veceta_1)$ has density $\Psi(x, \vecz)$.

(ii) Under the hypotheses of Theorem \ref{thm:main2}, for any $n\in\NN$,
\begin{equation}\label{Qn}
\vecq_n -\vecq_0 \Rightarrow \vecQ_n
\end{equation}
as $r\to 0$, where the random variable $(\xi_1,\veceta_1)$ has density $\Psi_0(x, \vecz)$.
\end{thm}

The main part of this paper is devoted to the proof of the following superdiffusive central limit theorem for the processes $\vecX_t$ and $\vecQ_n$, which in turn implies Theorems \ref{thm:main1} and \ref{thm:main2}. We will only assume that the random variable $(\xi_1,\veceta_1)$ is such that the marginal distribution of $\veceta_1$ is absolutely continuous on $\UB$ with respect to Lebesgue measure; there is no further assumption on the distribution of $\xi_1$. 
This hypothesis is satisfied for $(\xi_1,\veceta_1)$ with density $\Psi_0(x, \vecz)$, since
\begin{equation}
\overline\Psi_0(\vecz):=\int_0^\infty \Psi_0(x, \vecz)\, dx = \frac{1}{v_{d-1}} \int_{\RR_{>0}\times\UB}   \Psi_0( \vecz,x,\vecw) \, d x\,d\vecw= \frac{1}{v_{d-1}}.
\end{equation}
That is, the marginal distribution of  $\veceta_1$ is uniform on $\UB$.
We will later see that $(\xi_1,\veceta_1)$ with density $\Psi(x, \vecz)$ also complies with the above hypothesis (cf.~Proposition \ref{prop:main3}). The processes $\vecX_t$ and $\vecQ_n$ are independent of $\vecx_0$ and $\vecq_0$, respectively, and we will in the following fix $\vecv_0\in\US$. Also, the required assumptions on the scattering angle $\theta$ are significantly weaker than in the previous theorems.

\begin{thm}\label{thm:main3}
Let $d\geq 2$,  $\vecv_0\in\US$ and assume that the marginal distribution of $\veceta_1$ is absolutely continuous. Assume $\theta:[0,1)\to [-\pi,\pi]$ is measurable, so that
\begin{equation}\label{irr2}
\meas\{ w\in[0,1) : \theta(w) \notin \pi \QQ \} >0 .
\end{equation}
Then (i)
\begin{equation}\label{eq:main3}
\frac{\vecX_t}{\Sigma_d \sqrt{t\log t}}  \Rightarrow \scrN(0,I_d),
\end{equation}
and (ii)
\begin{equation}\label{eq:main4}
\frac{\vecQ_n}{\sigma_d \sqrt{n\log n}}  \Rightarrow \scrN(0,I_d).
\end{equation}
\end{thm}

In view of Theorem \ref{thm:MS},  Theorem \ref{thm:main3} implies Theorems \ref{thm:main1} and \ref{thm:main2}. \\

Statement (ii) in Theorem \ref{thm:MS} generalizes to the convergence of the random curve \eqref{thm:MS} obtained by linearly interpolating $\vecq_n$  \cite[Theorem 1.1]{Marklof:2011ho}. That is, under the conditions of Theorem \ref{plusthm:main2}, for $r\to 0$ and arbitrary fixed $n$,
\begin{equation}\label{pluseq:main234}
\vecY_{n,r}\Rightarrow \vecY_n 
\end{equation}
where the rescaled discrete-time limiting process is defined by
\begin{equation}
\vecY_n(t):= \frac{\vecQ_n(t)}{\sigma_d\sqrt{n\log n}} ,
\end{equation}
and 
\begin{equation}
\vecQ_n(t):=\vecQ_{\lfloor nt\rfloor} + \{ nt\}\, \xi_{\lfloor nt\rfloor +1}\vecV_{\lfloor nt\rfloor}
\end{equation}
denotes the linear interpolation of the discrete time displacements $\vecQ_0,\vecQ_1,\vecQ_2,\ldots$ We will prove in Section \ref{plussec:finite} that the $\vecY_n$ converges to $\vecW$ in finite-dimensional distribution. The last missing ingredient in the proof of Theorem \ref{plusthm:main2} is then the tightness of the probability measures associated with the sequence of random curves $(\vecY_n)_{n=1}^\infty$ in $\C_0([0,1])$, which is established in Section \ref{plussec:tightness}. Theorem \ref{plusthm:main1} follows from Theorem \ref{plusthm:main2} via estimates presented in Section \ref{plussec:221}. \\

It is interesting to compare the above results with the case of a random, rather than periodic, scatterer configuration, where the scatterers are placed at the points of a fixed realisation of a Poisson process in $\RR^d$. In the case of fixed scattering radius there is, to the best of our knowledge, no proof of a central limit theorem even in dimension $d=2$. In the Boltzmann-Grad limit, however, the work of Gallavotti \cite{Gallavotti1969}, Spohn \cite{Spohn:1978tt} and Boldrighini, Bunimovich and Sinai \cite{Boldrighini:1983jm} shows that we have an analogue of Theorem \ref{thm:MS}, where the limit random flight process $\vecX_t$ is governed by the linear Boltzmann equation. In this setting, \eqref{xieta} is a sequence of independent random variables, where $\xi_n$ has density $\Psi_0(x)=v_{d-1} \exp(-v_{d-1} x)$ and $\veceta_n$ is uniformly distributed in $\UB$. Routine techniques \cite{Papanicolaou:1975tn} show that in this case the central limit theorem holds for $\vecX_t$ with a standard $\sqrt t$ normalisation, and  for $\vecQ_n$ with a $\sqrt n$ normalisation.

\section{Outline of the proof of Theorem \ref{thm:main3}}\label{sec:outline}

We will now outline the central arguments in the proof of Theorem \ref{thm:main3} (ii) for discrete time by reducing the statement to four main lemmas, whose proof is given in Section \ref{sec:proof}. The continuous-time case (i) follows from (ii) via technical estimates supplied in Section \ref{sec:from}. We will assume from now on that $(\xi_1,\veceta_1)$ has density $\Psi_0(x, \vecz)$, and discuss the generalisation to more general distributions in Section \ref{sec:general}. We note that for $\veceta_0$ uniformly distributed in $\UB$,
\begin{equation}\label{eli}
\Psi_0(x, \vecz) = \Exp \Psi_0(\veceta_0,x,\vecz) ,
\end{equation}
and it is therefore equivalent to consider instead of \eqref{xieta} the Markov chain 
\begin{equation}\label{xieta1}
\big( (\xi_n,\veceta_n) \big)_{n=0}^\infty
\end{equation}
with the same transition probability \eqref{transprob}, $\veceta_0$ uniformly distributed in $\UB$ and $\xi_0=0$.
The sequence
\begin{equation}
\uveceta = \big(\veceta_n\big)_{n=0}^\infty ,
\end{equation}
with $\veceta_0$ as defined above, is itself generated by a Markov chain on the state space $\UB$ with transition probability
\begin{equation}\label{Markov1}
\Prob\big( \veceta_n\in\scrA \bigm| \veceta_{n-1} \big) 
= \int_\scrA K_0(\veceta_{n-1},\vecz) \, d\vecz
\end{equation}
where
\begin{equation}\label{K0def}
K_0(\vecw,\vecz) := \int_0^\infty \Psi_0( \vecw,x,\vecz) d x .
\end{equation}

The objective is to prove a central limit theorem of sums of the random variables $\xi_n \vecV_{n-1}$. The first observation is that these are of course not independent. If we, however, condition on the sequence $\uveceta$, then the $\vecV_n$ are deterministic, and $(\xi_n)_{n=1}^\infty$ is a sequence of independent (but not identically distributed) random variables, 
\begin{equation}
\Prob\big(\xi_n\in (x,x+dx) \bigm| \uveceta\big) = \frac{\Psi_0(\veceta_{n-1},x,\veceta_n)\, dx}{K_0(\veceta_{n-1},\veceta_n)}.
\end{equation}
The plan is now to apply the Lindeberg central limit theorem to the sum of independent random variables, $\vecQ_n=\sum_{j=1}^n \xi_j \vecV_{j-1}$, conditioned on $\uveceta$. 

To this end we first truncate $\vecQ_n$ by defining the random variable
\begin{equation}
\vecQ_n' := \sum_{j=1}^n \xi_j' \vecV_{j-1}  
\end{equation}
with
\begin{equation}\label{xiprime}
\xi_j' := \xi_j \ind{\xi_j^2 \leq   j (\log j)^\gamma}
\end{equation}
for some fixed $\gamma\in(1,2)$. The following lemma tells us that it is sufficient to prove Theorem \ref{thm:main3} (ii) for $\vecQ_n'$ instead of $\vecQ_n$.

\begin{lem}\label{mainlem0}
We have
\begin{equation}
\sup_{n\in\NN}\|\vecQ_n-\vecQ_n'\| <\infty
\end{equation}
almost surely.
\end{lem}

To prove the central limit theorem for $\vecQ_n'$, we center $\xi_j'$ by setting 
\begin{equation}
\txi_j = \xi_j' -m_j,  
\end{equation}
with the conditional expectation
\begin{equation}\label{conExp}
m_j := \Exp\big( \xi_j' \bigm| \uveceta\big) = \frac{K_{1,r_j}(\veceta_{j-1},\veceta_j)}{K_0(\veceta_{j-1},\veceta_j)}
\end{equation}
where $r_j:=\sqrt{j(\log j)^\gamma}$ and
\begin{equation}\label{K1rdef}
K_{1,r}(\vecw,\vecz) := \int_0^r  x \Psi_0( \vecw,x,\vecz) d x .
\end{equation}
Let
\begin{equation}\label{2sums}
\widetilde\vecQ_n := \sum_{j=1}^n \txi_j \vecV_{j-1} .
\end{equation}

The following lemma shows that $\vecQ_n'$ and $\widetilde\vecQ_n$ are close relative to $\sqrt{n\log n}$.

\begin{lem}\label{mainlem1}
The sequence of random variables 
\begin{equation}
\frac{\vecQ_n'-\widetilde\vecQ_n}{\sqrt{n \log\log n}} 
\end{equation}
is tight if $d=2$, and  
\begin{equation}
\frac{\vecQ_n'-\widetilde\vecQ_n}{\sqrt{n}} 
\end{equation}
is tight if $d\geq 3$.
\end{lem}

It is therefore sufficient to prove Theorem \ref{thm:main3} (ii) for $\widetilde\vecQ_n$ in place of $\vecQ_n$. This will be achieved by applying the Lindeberg central limit theorem to the conditional sum as aluded to above. 
We begin by estimating the conditional variance. Set
\begin{equation}\label{cona}
a_j^2:= \Var\big(\txi_j \bigm| \uveceta\big)= \frac{K_{2,r_j}(\veceta_{j-1},\veceta_j)}{K_0(\veceta_{j-1},\veceta_j)} - m_j^2 ,
\end{equation}
with
\begin{equation}
K_{2,r}(\vecw,\vecz) := \int_0^r  x^2 \Psi_0( \vecw,x,\vecz) d x .
\end{equation}

\begin{lem}\label{mainlem2}
There is a constant $\sigma_d>0$ such that, for $n\to\infty$,
\begin{equation}\label{eq:mainlem2}
\frac{\Exp\big( \widetilde\vecQ_n \otimes\widetilde\vecQ_n\bigm| \uveceta \big)}{n\log n}  = \frac{\sum_{j=1}^n  a_j^2 \vecV_{j-1} \otimes\vecV_{j-1}}{n\log n}\toprob  \sigma_d^2 \,  I_d .
\end{equation}
\end{lem}

By taking the trace in \eqref{eq:mainlem2}, we have in particular
\begin{equation}\label{eq:mainlem2b}
\frac{A_n^2}{n\log n}\toprob  d\,\sigma_d^2 
\end{equation}
for
\begin{equation}\label{conA}
A_n^2 :=\sum_{j=1}^n a_j^2 = \Exp\big( \|\widetilde\vecQ_n\|^2 \bigm| \uveceta \big).
\end{equation}
Recall that convergence in probability $x_n\toprob x$ is defined as $\lim_{n\to\infty} \Prob(|x_n-x|>\epsilon)=0$ for any $\epsilon>0$.

The next lemma verifies the Lindeberg conditions for random $\uveceta$.

\begin{lem}\label{mainlem3}
For any fixed $\vareps>0$,
\begin{equation}
\label{Lindeberg_condition}
A_n^{-2}
\sum_{j=1}^n
\condexpect{\txi_j^2 \ind{\txi_j^2 > \vareps^2 A_n^2}}{\uveceta}
\toprob 0
\end{equation}
as $n\to\infty$.
\end{lem}

Given these lemmas, let us now conclude the proof of the fact that 
\begin{equation}\label{eq:main4b}
\vecY_n:=\frac{\widetilde\vecQ_n}{\sigma_d \sqrt{n\log n}}  \Rightarrow \scrN(0,I_d).
\end{equation}
By Chebyshev's inequality we have, for any $K>0$,
\begin{equation}
\Prob\big( \|\vecY_n\| > K \bigm| \uveceta \big) \leq \frac{1}{K^2} \Exp\big( \|\vecY_n\|^2 \bigm| \uveceta \big) ,
\end{equation}
and thus, for any $\kappa>0$,
\begin{equation}\label{Chebyshev}
\begin{split}
\Prob\big( \|\vecY_n\| > K \big) 
& \leq \frac{\kappa^2}{K^2} + \Prob\big(\Exp\big( \|\vecY_n\|^2 \bigm| \uveceta \big) >\kappa^2\big) \\
& = \frac{\kappa^2}{K^2} + \Prob\big(A_n^2 > \kappa^2 \sigma_d^2 \; n\log n \big) .
\end{split}
\end{equation}
By \eqref{eq:mainlem2b}, the second term on the right hand side of \eqref{Chebyshev} converges to $0$ as $n\to\infty$, if we choose $\kappa=d$, say.
So \eqref{Chebyshev} implies that the sequence of random variables $\vecY_n$ is tight. By the Helly-Prokhorov theorem, there is an infinite subset $S_1\subset \NN$ so that $\vecY_n$ converges in distribution along $n\in S_1$ to some limit $\vecY$. Assume for a contradiction that $\vecY$ is {\em not} distributed according to $\scrN(0,I_d)$. The Borel-Cantelli lemma implies that there is an infinite subset $S_2\subset S_1$, so that in the statements of Lemmas \ref{mainlem2} and \ref{mainlem3} we have almost-sure convergence along $n\in S_2$:
\begin{equation}\label{eq:mainlem222}
\frac{\Exp\big( \widetilde\vecQ_n \otimes\widetilde\vecQ_n\bigm| \uveceta \big)}{n\log n}  \toas  \sigma_d^2\,  I_d ,
\end{equation}
\begin{equation}\label{eq:mainlem2bbb}
\frac{A_n^2}{n\log n}\toas  d\, \sigma_d^2 ,
\end{equation}
and
\begin{equation}
\label{Lindeberg_condition22}
A_n^{-2}
\sum_{j=1}^n
\condexpect{\txi_j^2 \ind{\txi_j^2 > \vareps^2 A_n^2}}{\uveceta} 
\toas 0 .
\end{equation}
The hypotheses of the Lindeberg central limit theorem are met, and we infer that $\vecY_n\Rightarrow \scrN(0,I_d)$ for $n\in S_2$. (We use the Lindeberg theorem for {\em triangular arrays} of independent random variables, since we have verified the Lindeberg conditions only along a subsequence.) This, however, contradicts our assumption that $\vecY$ is not normal, and hence $\scrN(0,I_d)$ is indeed the unique limit point of any converging subsequence. This in turn implies that every sequence converges, and therefore completes the proof of \eqref{eq:main4b}. In view of Lemmas \ref{mainlem0} and \ref{mainlem1}, this implies Theorem \ref{thm:main3} (ii) (still under the assumption that $(\xi_1,\veceta_1)$ has density $\Psi_0(x,\vecz)$).

Let us briefly describe the further contents of this paper. In Section \ref{sec:transition} we recall the basic properties of the transition kernel $\Psi_0(\vecw,x,\vecz)$ from  \cite{Marklof:2011di}. Section \ref{sec:moments} establishes key estimates for the moments $K_{p,r}(\vecw,\vecz)$, $m_j$ and $a_j$ introduced above. In Sections \ref{sec:spectral} and \ref{sec:exponential} we prove spectral gap estimates and exponential mixing for the discrete time Markov process defined in \eqref{Markov1}. The estimates from Sections \ref{sec:moments}--\ref{sec:exponential} are the main input in the proof of Lemmas \ref{mainlem0}--\ref{mainlem3}, which is given in Section \ref{sec:proof}. In Section \ref{sec:general} we show that the discrete-time statement in Theorem \ref{thm:main3} (ii) holds for more general initial distributions than $\Psi_0(x,\vecz)$. It holds in particular for $\Psi(x,\vecz)$, which appears in the continuous-time variant. Section \ref{sec:from} explains how to pass from discrete to continuous time, thus completing the proof of Theorem \ref{thm:main3} (i).

\section{The transition kernel}\label{sec:transition}

In dimension $d=2$ we have the following explicit formula for the transition kernel. For $w,z\in(-1,1)$, 
\begin{equation}\label{Xp}	
\Psi_0( w,x,z)=\frac{6}{\pi^2}
\begin{cases}
\displaystyle
\ind{0\le  x<\frac{1}{1+w}}+\ind{\frac{1}{1+w}\le  x < \frac{1}{1+z}}\frac{x^{-1}-1-z}{w-z}
&\text{ if }
0<w, \ -w< z < w,
\\[15pt]
\displaystyle
\ind{0\le  x<\frac{1}{1+z}}+\ind{\frac{1}{1+z}\le  x < \frac{1}{1+w}}\frac{x^{-1}-1-w}{z-w}
&\text{ if }
0<z, \ -z< w < z,
\\[15pt]
\displaystyle
\ind{0\le  x<\frac{1}{1-w}}+\ind{\frac{1}{1-w}\le  x < \frac{1}{1-z}}\frac{x^{-1}-1+z}{z-w}
&\text{ if }
w<0, \ w< z <-w,
\\[15pt]
\displaystyle
\ind{0\le  x<\frac{1}{1-z}}+\ind{\frac{1}{1-z}\le  x < \frac{1}{1-w}}\frac{x^{-1}-1+w}{w-z}
&\text{ if }
z<0, \ z< w <-z.
\end{cases}
\end{equation}
This formula has been derived, independently and with different methods, by Marklof and Str\"ombergsson \cite{Marklof:2008dr}, Caglioti and Golse \cite{Caglioti:2008um,Caglioti:2010ir} and by Bykovskii and Ustinov \cite{Bykovskiui:2009db}.

In dimension $d\geq 3$ we have no such explicit formulas for the transition kernel. We recall from \cite{Marklof:2010ib,Marklof:2011ho,Marklof:2011di} the following properties.
If $d\geq 3$, the function
\begin{equation}
	\Psi_0: \scrB_1^{d-1}\times \RR_{>0}\times  \scrB_1^{d-1} \to [0,1]
\end{equation}
is continuous. $\Psi_0( \vecw,x,\vecz)$ depends only on 
$ x$, $w:=\|\vecw\|$, $z:=\|\vecz\|$ and the angle $\varphi:=\varphi(\vecw,\vecz)\in[0,\pi]$ between the vectors $\vecw,\vecz\in\UB$. 
Note that in dimension $d=2$ the angle $\varphi$ can only take the values $0$ and $\pi$. For statements that are specific to dimension $d=2$, we will often use $w\in(-1,1)$ instead of $\vecw$, and $|w|$ instead of $w=\|\vecw\|$. We recall once more that $\Psi_0( \vecw,x,\vecz) = \Phi_0( x,\vecw,-\vecz)$ in the notation of \cite{Marklof:2010ib,Marklof:2011ho,Marklof:2011di}, and so in particular the angle $\varphi$ between $\vecw,\vecz$ becomes $\pi-\varphi$.

Our proofs will exploit the following estimates on the transition kernel \cite{Marklof:2011di,Strombergsson:2011ex}.
All bounds are uniform in $x>0$ and $\vecw,\vecz\in\scrB_1^{d-1}$.
We have by \cite[Thm.\ 1.1]{Marklof:2011di},
\begin{equation}\label{PHI0ZEROSMALLTHMRES}
\frac{1-2^{d-1}v_{d-1}  x}{\zeta(d)}\leq\Psi_0( \vecw,x,\vecz)\leq\frac{1}{\zeta(d)} .
\end{equation}
Furthermore, by \cite[Thm.\ 1.7]{Marklof:2011di}, there exists a continuous and uniformly bounded function
$F_{0,d}:\RR_{>0}\times\RR_{>0}\times\RR_{\geq0}\to\RR_{\geq0}$
such that
\begin{equation}\label{PHI0XILARGETHMRES}
\Psi_0( \vecw,x,\vecz)=
 x^{-2+\frac 2d}F_{0,d}\Bigl( x^{\frac 2d}(1-z), x^{\frac 2d}(1-w),
 x^{\frac1d}(\pi-\varphi)\Bigr)
+O(E),
\end{equation}
where the error term is
\begin{equation}\label{PHI0XILARGETHMEDEF}
E=\begin{cases}
 x^{-2} &\text{if }\:d=2,
\\
 x^{-2}\log(2+\min( x,(\pi-\varphi)^{-1}))&\text{if }\:d=3,
\\
\min\bigl( x^{-2}, x^{-3+\frac2{d-1}}(\pi-\varphi)^{2-d+\frac2{d-1}}\bigr)
&\text{if }\:d\geq4.\end{cases}
\end{equation}

It is noted in \cite{Marklof:2011di} that $F_{0,d}(t_1,t_2,\alpha)$ is uniformly
bounded from below for $t_1,t_2,\alpha$ near zero. That is,
there is a small constant $c>0$ which only depends on $d$ such that
\begin{equation}\label{F0DBOUNDBELOW1}
\max(t_1,t_2,\alpha)<c
\:\Longrightarrow\: F_{0,d}(t_1,t_2,\alpha)>c.
\end{equation}
Furthermore, the support of $F_{0,d}$ is contained in $(0,c']\times(0,c']\times\RR_{\geq0}$ for some $c'>0$, and for any fixed $t_1,t_2>0$, the function $F_{0,d}(t_1,t_2,\cdot)$ has compact support.

In dimension $d\geq 3$, the following upper bound will prove useful \cite[Thm.\ 1.8]{Strombergsson:2011ex}:
\begin{equation}\label{CYLINDER2PTSMAINTHMRES}
\Psi_0( \vecw,x,\vecz)\ll\begin{cases}
 x^{-2}\min\Bigl\{1,
( x\varphi^{d-2})^{-1+\frac2{d-1}}\Bigr\}
&\text{if }\:\varphi\in[0,\frac\pi2]\\
 x^{-2+\frac2d}\min\Bigl\{1,
( x(\pi-\varphi)^{d})^{-1+\frac2{d(d-1)}}\Bigr\}
&\text{if }\:\varphi\in[\frac\pi2,\pi] .
\end{cases}
\end{equation}

The notation $f\ll g$ is here defined as $f=O(g)$, i.e., there exists a constant $C>0$ such that $|f|\leq C |g|$. The notation $f\asymp g$ used below means that $g\ll f \ll g$, i.e., there exist a constant $C\geq 1$ such that $C^{-1} |g|\leq |f|\leq C |g|$.

The support of $\Psi_0( \vecw,x,\vecz)$ is described by a continuous function 
$ x_0:\UB\times\UB\to\RR_{>0}$. We have
$\Psi_0( \vecw,x,\vecz)>0$ holds if and only if 
$ x< x_0(\vecw,\vecz)$. 
Set \begin{equation}
t:=t(\vecw,\vecz):=\max( 1-w,1-z)\in(0,1].
\end{equation}
If $d\geq 3$, then \cite[Prop.\ 1.9]{Strombergsson:2011ex} tells us that
\begin{equation}\label{PHI0SUPPORTBOUNDS}
 x_0(\vecw,\vecz)\asymp
\begin{cases}
\max(t^{-\frac{d-2}2},t^{-\frac{d-1}2}\varphi)
&\text{if }\varphi\in[0,\frac\pi2]\\
\min(t^{-\frac d2},t^{-\frac{d-1}2}(\pi-\varphi)^{-1}) 
&\text{if }\varphi\in[\frac\pi2,\pi].
\end{cases}
\end{equation}
(If $\varphi=\pi$ then the right hand side of \eqref{PHI0SUPPORTBOUNDS} 
should be interpreted as $t^{-\frac d2}$.)

For the distribution of free path length between consecutive collisions,
\begin{equation}
\Psi_0( x)=\frac{1}{v_{d-1}} \int_\UB\int_\UB \Psi_0( \vecw,x,\vecz)\, d\vecw\,d\vecz ,
\end{equation}
we have the following tail estimate  \cite[Theorem 1.14]{Marklof:2011di} ($\Psi_0$ is denoted $\overline{\Phi}_0$ in \cite{Marklof:2011di}):
For $x\to\infty$,
\begin{equation}\label{PHIBARXILARGETHMRES1}
\Psi_0(x)=
\Theta_d \, x^{-3}
+O\big(x^{-3-\frac 2d}\big)
\times \begin{cases}
1&\text{if }\:d=2 
\\
\log x &\text{if }\:d=3
\\
1&\text{if }\:d\geq4
\end{cases}
\end{equation}
with
\begin{equation}
\Theta_d:=\frac{2^{2-d}}{d(d+1)\zeta(d)}.
\end{equation}
This asymptotic estimate sharpens earlier upper and lower bounds by Bourgain, Golse and Wennberg \cite{Bourgain:1998gu,Golse:2000ka}. Note that the variances in Theorems \ref{thm:main1} and \ref{thm:main2} are related to the above tail via
\begin{equation}
\Sigma_d^2 = \frac{\Theta_d}{2d \xibar}, \qquad \sigma_d^2 = \frac{\Theta_d}{2d}.
\end{equation}

\section{Moment estimates}\label{sec:moments}

We now provide key estimates of the random variables introduced in the previous section. 
For $p=0,1,2$ and $r>0$, set
\begin{equation}
K_{p,r}(\vecw,\vecz) := \int_0^r  x^p \Psi_0( \vecw,x,\vecz) d x 
\end{equation}
and
\begin{equation}\label{Kpdef}
K_p(\vecw,\vecz) := \int_0^\infty  x^p \Psi_0( \vecw,x,\vecz) d x .
\end{equation}
Note that $v_{d-1}^{-1} K_0(\vecw,\vecz) \, d\vecw\,d\vecz$ defines a probability measure on $\UB\times\UB$.
We furthermore define the random variables (recall Section \ref{sec:outline})
\begin{equation}\label{muj}
m_j:=\Exp\big( \xi_j' \bigm| \uveceta\big) =\frac{K_{1,r_j}(\veceta_{j-1},\veceta_j)}{K_0(\veceta_{j-1},\veceta_j)},
\qquad
\mu_j:=\Exp\big( \xi_j \bigm| \uveceta\big) =\frac{K_1(\veceta_{j-1},\veceta_j)}{K_0(\veceta_{j-1},\veceta_j)},
\end{equation}
\begin{equation}\label{betaj}
b_j^2:=\Exp\big( {\xi_j'}^2  \bigm| \uveceta\big) = \frac{K_{2,r_j}(\veceta_{j-1},\veceta_j)}{K_0(\veceta_{j-1},\veceta_j)},
\qquad
\beta_j^2:=\Exp\big( {\xi_j}^2  \bigm| \uveceta\big) = \frac{K_{2}(\veceta_{j-1},\veceta_j)}{K_0(\veceta_{j-1},\veceta_j)},
\end{equation}
\begin{equation}\label{alphaj}
 a_j^2:=\Var\big( \xi_j'  \bigm| \uveceta\big) = b_j^2 - m_j^2 ,
\qquad
 \alpha_j^2:=\Var\big( \xi_j  \bigm| \uveceta\big) = \beta_j^2 - \mu_j^2 ,
\end{equation}
and
\begin{equation}
A_n^2 :=\sum_{j=1}^n a_j^2,
\end{equation}
with $r_j=\sqrt{j(\log j)^\gamma}$ for some fixed $\gamma\in(1,2)$.

\begin{lem}\label{six-one}
Let $d=2$. Then, for $w,z\in(-1,1)$,
\begin{equation}
\label{K0}
K_0(w,z)
=
\frac{6}{\pi^2} 
\begin{cases}
\displaystyle
\frac{1}{w-z}\ln \frac{1+w}{1+z}
&\text{ if }
w+z\ge0,
\\[15pt]
\displaystyle
\frac{1}{z-w}\ln \frac{1-w}{1-z}
&\text{ if }
w+z\le0,
\end{cases}
\end{equation}

\begin{equation}
\label{K1}
K_1(w,z)
=
\frac{3}{\pi^2} 
\begin{cases}
\displaystyle
\frac{1}{(1+w)(1+z)}
&\text{ if }
w+z\ge0,
\\[15pt]
\displaystyle
\frac{1}{(1-w)(1-z)}
&\text{ if }
w+z\le0,
\end{cases}
\end{equation}
\begin{equation}
\label{K2}
K_2(w,z)
=
\frac{1}{\pi^2} 
\begin{cases}
\displaystyle
\frac{2+w+z}{(1+w)^2(1+z)^2}
&\text{ if }
w+z\ge0,
\\[15pt]
\displaystyle
\frac{2-w-z}{(1-w)^2(1-z)^2}
&\text{ if }
w+z\le0.
\end{cases}
\end{equation}
\end{lem}

\begin{proof}
These follow from the explicit formula \eqref{Xp} by direct computation.
\end{proof}

\begin{lem}\label{K0upperlower}
Let $d\geq 3$. Then
\begin{equation}\label{K0lower}
\inf_{\vecw,\vecz\in\UB} K_0(\vecw,\vecz) \geq \frac{1}{2^d v_{d-1} \zeta(d)} >0 ,
\end{equation}
\begin{equation} \label{K0upper}
\sup_{\vecw,\vecz\in\UB} K_0(\vecw,\vecz) < \infty.
\end{equation}
\end{lem}

\begin{proof}
We have
\begin{equation}
K_0(\vecw,\vecz)\geq \int_0^y \Psi_0( \vecw,x,\vecz) d x
\end{equation}
for any $y\geq 0$. Theorem 1.1 in \cite{Marklof:2011di} states that for $ x>0$ and $\vecw,\vecz\in\UB$,
\begin{equation}
\Psi_0( \vecw,x,\vecz) \geq \frac{1-2^{d-1}  v_{d-1}  x}{\zeta(d)} ,
\end{equation}
and the lower bound follows with the choice $y=(2^{d-1}  v_{d-1})^{-1}$. The upper bound follows from \eqref{CYLINDER2PTSMAINTHMRES} which tells us that
\begin{equation}\label{andreasbound}
\Psi_0( \vecw,x,\vecz) = O( x^{-2+\frac2d}) ,
\end{equation}
where the implied constant is independent of $x,\vecw,\vecz$. 
\end{proof}

\begin{lem}\label{K1upperlower}
Let $d\geq 3$. For $\vecw,\vecz\in\UB$,
\begin{equation} \label{K1lower}
K_1(\vecw,\vecz) \gg \min( t^{-1} , (\pi-\varphi)^{-2} ) ,
\end{equation}
and
\begin{equation} \label{K1upper}
K_1(\vecw,\vecz) \ll
\begin{cases}
1 + \log \max( t^{-\frac{d-2}{2}}, t^{-\frac{d-1}{2}} \varphi) & \text{if $\varphi\in [0,\frac{\pi}{2}]$,} \\
\min( t^{-1} , t^{-1+\frac1d} (\pi-\varphi)^{-\frac2d}) & \text{if $\varphi\in [\frac{\pi}{2},\pi]$.}
\end{cases}
\end{equation}
\end{lem}

\begin{proof}
As to the lower bound \eqref{K1lower}, we note that by \eqref{PHI0XILARGETHMRES}
\begin{equation}
K_1(\vecw,\vecz) \geq \int\limits_{\substack{ x>0\\  x^{\frac 2d} t(\vecw,\vecz) < c \\  x^{\frac1d}(\pi-\varphi)<c}}
 x \Bigl\{  x^{-2+\frac 2d}F_{0,d}\Bigl( x^{\frac 2d}(1-z), x^{\frac 2d}(1-w),
 x^{\frac1d}(\pi-\varphi)\Bigr)
- |O(E)| \Bigr\} d x ,
\end{equation}
which, in view of \eqref{PHI0XILARGETHMEDEF} and \eqref{F0DBOUNDBELOW1}, implies
\begin{equation}
\begin{split}
K_1(\vecw,\vecz) & >  c \int\limits_{\substack{ x>0\\  x^{\frac 2d} t(\vecw,\vecz) < c \\  x^{\frac1d}(\pi-\varphi)<c}}
\Bigl\{  x^{-1+\frac 2d}  - |O( (1+x)^{-1} \log(2+ x) )| \Bigr\} d x \\
& \gg \min\big( t(\vecw,\vecz)^{-1}, (\pi-\varphi)^{-2} \big).
\end{split}
\end{equation}

The upper bound \eqref{K1upper} follows from \eqref{CYLINDER2PTSMAINTHMRES} and \eqref{PHI0SUPPORTBOUNDS}: for $\varphi\in [0,\frac{\pi}{2}]$, 
\begin{equation}
K_1(\vecw,\vecz) \ll 1+\int_1^{ x_0(\vecw,\vecz)}  x^{-1} d x \ll 1+ \log \max( t^{-\frac{d-2}{2}}, t^{-\frac{d-1}{2}} \varphi)  ,
\end{equation}
and for $\varphi\in [\frac{\pi}{2},\pi]$, we have
\begin{equation}
K_1(\vecw,\vecz) \ll \int_0^{ x_0(\vecw,\vecz)}  x^{-1+\frac2d} d x \ll \min( t^{-1} , t^{-1+\frac1d} (\pi-\varphi)^{-\frac2d}) .
\end{equation}
\end{proof}

\begin{lem}\label{K2upperlower}
Let $d\geq 3$. Then, for $\vecw,\vecz\in\UB$,
\begin{equation} \label{K2lower}
K_2(\vecw,\vecz) \gg \min( t^{-(1+\frac{d}{2})} , (\pi-\varphi)^{-(d+2)} ) ,
\end{equation}
and
\begin{equation} \label{K2upper}
K_2(\vecw,\vecz) \ll
\begin{cases}
\max( t^{-\frac{d-2}{2}}, t^{-\frac{d-1}{2}} \varphi) & \text{if $\varphi\in [0,\frac{\pi}{2}]$,} \\
\min( t^{-(1+\frac{d}{2})} , t^{-\frac{d+1}{2}+\frac1d} (\pi-\varphi)^{-(1+\frac2d)}) & \text{if $\varphi\in [\frac{\pi}{2},\pi]$.}
\end{cases}
\end{equation}
\end{lem}

\begin{proof}
The lower bound \eqref{K2lower} follows from
\begin{equation}
K_2(\vecw,\vecz) \geq \int\limits_{\substack{ x>0\\  x^{\frac 2d} t(\vecw,\vecz) < c \\  x^{\frac1d}(\pi-\varphi)<c}}
 x^2 \Bigl\{  x^{-2+\frac 2d}F_{0,d}\Bigl( x^{\frac 2d}(1-z), x^{\frac 2d}(1-w),
 x^{\frac1d}(\pi-\varphi)\Bigr)
- |O(E)| \Bigr\} d x ,
\end{equation}
and \eqref{PHI0XILARGETHMEDEF} and \eqref{F0DBOUNDBELOW1}. Hence
\begin{equation}
\begin{split}
K_2(\vecw,\vecz) & >  c \int\limits_{\substack{ x>0\\  x^{\frac 2d} t(\vecw,\vecz) < c \\  x^{\frac1d}(\pi-\varphi)<c}}
\Bigl\{  x^{\frac 2d}  - |O(\log(2+ x) )| \Bigr\} d x \\
& \gg \min\big( t(\vecw,\vecz)^{-(1+\frac{d}{2})}, (\pi-\varphi)^{-(d+2)} \big).
\end{split}
\end{equation}

The upper bound \eqref{K2upper} follows from \eqref{CYLINDER2PTSMAINTHMRES} and \eqref{PHI0SUPPORTBOUNDS}: for $\varphi\in [0,\frac{\pi}{2}]$, 
\begin{equation}\label{up1}
K_2(\vecw,\vecz) \ll \int_0^{ x_0(\vecw,\vecz)} d x \ll \max( t^{-\frac{d-2}{2}}, t^{-\frac{d-1}{2}} \varphi)  ,
\end{equation}
and for $\varphi\in [\frac{\pi}{2},\pi]$, we have
\begin{equation}\label{up2}
K_2(\vecw,\vecz) \ll \int_0^{ x_0(\vecw,\vecz)}  x^{\frac2d} d x \ll \min( t^{-(1+\frac{d}{2})} , t^{-\frac{d+1}{2}+\frac1d} (\pi-\varphi)^{-(1+\frac2d)}) .
\end{equation}
\end{proof}

\begin{prop}\label{prop:tail2}
Let $d=2$. For $u \to\infty$,
\begin{equation}\label{masy}
\Prob(\mu_j> u) \sim \frac{3}{4\pi^2} \, \frac{1}{u^2\log u} .
\end{equation}
\end{prop}

\begin{proof}
By the invariance of the integrand under $(w,z)\mapsto -(w,z)$ and $(w,z)\mapsto(z,w)$ we have
\begin{equation}
\begin{split}
\Prob(\mu_j>u)
&=\frac{1}{2} \int_{(-1,1)^2} \ind{K_1(w,z)>u K_0(w,z)} K_0(w,z)\, dw\, dz \\
&=2 \int_{|z|<w<1} \ind{K_1(w,z)>u K_0(w,z)} K_0(w,z)\, dw\, dz .
\end{split}
\end{equation}
(The factor $\frac12=v_1^{-1}$ corresponds to the normalization in the remark following \eqref{Kpdef}.)
In this range of integration, we have explicitly
\begin{equation}
K_0(w,z)=\frac{6}{\pi^2} \;\frac{1}{w-z}\ln \frac{1+w}{1+z}
\end{equation}
and
\begin{equation}
\frac{K_0(w,z)}{K_1(w,z)}
=
2\; \frac{(1+w)(1+z)}{w-z}\ln \frac{1+w}{1+z}.
\end{equation}
Using the variable substitution
\begin{equation}\label{wxyz}
x=w-z,\qquad y = \frac{1+z}{w-z},
\end{equation}
we have, with the shorthand $f(y)=2 y(1+y) \ln(1+y^{-1})$,
\begin{equation}
\begin{split}
\Prob(\mu_j>u)
= &\frac{12}{\pi^2} \int_{0}^\infty \int_{\frac{2}{1+2y}}^{\frac{2}{1+y}} \ind{x f(y)<u^{-1}} \ln(1+y^{-1})\; dx\, dy \\
=&\frac{12}{\pi^2} \int_{0}^\infty \ind{\frac{2 f(y)}{1+y}<u^{-1}} \left( \frac{2}{1+y}-\frac{2}{1+2y}\right)  \ln(1+y^{-1}) \; dy\\ 
& +\frac{12}{\pi^2} \int_{0}^\infty \ind{\frac{2 f(y)}{1+2y}<u^{-1}<\frac{2 f(y)}{1+y}} \left( \frac{1}{uf(y)}-\frac{2}{1+2y}\right)  \ln(1+y^{-1}) \; dy .
\end{split}
\end{equation}
The first term equals
\begin{multline}
\frac{12}{\pi^2} \int_{0}^\infty \ind{4y\ln(1+y^{-1})<u^{-1}} \left(2y+O(y^2)\right)  \ln(1+y^{-1}) \; dy\\ 
 = \frac{12}{\pi^2} \left(y^2 \ln(1+y^{-1}) + O(y^2) \right)  \bigg|_{0_+}^{4y\ln(1+y^{-1})=u^{-1}} = \frac{3}{4\pi^2}\; \frac{1}{u^2 \ln u} (1+o(1)).
\end{multline}
The second term is bounded above by
\begin{equation}
\begin{split}
& \frac{12}{\pi^2} \int_{0}^\infty \ind{\frac{2 f(y)}{1+2y}<u^{-1}<\frac{2 f(y)}{1+y}} \left( \frac{2}{1+y}-\frac{2}{1+2y}\right) \ln(1+y^{-1}) \; dy \\
& = \frac{12}{\pi^2} \left(y^2 \ln(1+y^{-1}) + O(y^2) \right)  \bigg|_{\frac{2 f(y)}{1+y}=u^{-1}}^{\frac{2 f(y)}{1+2y}=u^{-1}} 
= o\bigg( \frac{1}{u^2 \ln u} \bigg)  .
\end{split}
\end{equation}
This proves \eqref{masy}. 
\end{proof}

\begin{prop}\label{prop:tail3}
Let $d\geq 3$. There are constants $c_2> c_1>0$, such that for $u\geq 1$,
\begin{equation}\label{00mupper}
c_1 u^{-(1+\frac{d}{2})} \leq \Prob(\mu_j> u) \leq c_2 u^{-(1+\frac{d}{2})} .
\end{equation}
\end{prop}

\begin{proof}
The upper bounds in \eqref{K1upper} for $K_1$, and the upper and lower bounds for  $K_0$ in Lemma \ref{K0upperlower} imply 
\begin{equation}\label{00st1}
\int_{K_1> u K_0} K_0 \, d\vecw\,d\vecz 
\ll \int\limits_{\substack{\log \max( t^{-\frac{d-2}{2}}, t^{-\frac{d-1}{2}} \varphi) \gg u\\ 0 \leq \varphi\leq \frac{\pi}{2}}} d\vecw\,d\vecz
+  \int\limits_{\substack{\min( t^{-1} , t^{-1+\frac1d} (\pi-\varphi)^{-\frac2d}) \gg u\\ \frac{\pi}{2} \leq \varphi\leq \pi}} d\vecw\,d\vecz  .
\end{equation}
Using spherical coordinates, we see that the second term is, up to a multiplicative constant, equal to
\begin{equation}
\int\limits_{\substack{ 0<z<w<1, \; \frac{\pi}{2} \leq \varphi\leq \pi \\ \min( (1-z)^{-1} , (1-z)^{-1+\frac1d} (\pi-\varphi)^{-\frac2d})\gg u}} w^{d-2} dw\, z^{d-2} dz \, (\sin\varphi)^{d-3} d\varphi .
\end{equation}
We now substitute $x=1-w$, $y=1-z$, $\phi=\pi-\varphi$, and integrate $x$ over $(0,y)$, and then over $y$. This yields (again up to a multiplicative constant)
\begin{equation}
\begin{split}
& \int_0^{\pi/2}
\min\big(u^{-1}, u^{-\frac{d}{d-1}} \phi^{-\frac{2}{d-1}} \big)^2 (\sin\phi)^{d-3} d\phi \\
& = u^{-2} \int_0^{u^{-1/2}} (\sin\phi)^{d-3} d\phi +
u^{-\frac{2d}{d-1}}  \int_{u^{-1/2}}^{\pi/2} \phi^{-\frac{4}{d-1}}  (\sin\phi)^{d-3} d\phi \\
& \ll u^{-(1+\frac{d}{2})} .
\end{split}
\end{equation}
The first term in \eqref{00st1} can be dealt with similarly, and yields a lower order contribution. This proves the upper bound 
\begin{equation}
\Prob(\mu_j> u) \ll u^{-(1+\frac{d}{2})} .
\end{equation}
As to the lower bound in \eqref{00mupper}, we use \eqref{K1lower} (and the same variable substitutions as above):
\begin{equation}\label{00st11}
\begin{split}
\int_{K_1> u K_0} K_0 \, d\vecw\,d\vecz 
& \gg \int\limits_{\substack{\min( t^{-1} ,(\pi-\varphi)^{-2}) \gg u\\ \frac{\pi}{2} \leq \varphi\leq \pi}} d\vecw\,d\vecz  \\
& \gg u^{-2} \int_0^{u^{-1/2}} (\sin\phi)^{d-3} d\varphi \\ 
& \gg u^{-(1+\frac{d}{2})} .
\end{split}
\end{equation}
\end{proof}

%A little more work shows that we in fact have the asymptotics 
%\begin{equation}
%\Prob(\mu_j> u)\sim \iota_d\, u^{-(1+\frac{d}{2})}
%\end{equation}
%for some constant $\iota_d>0$, akin to the two-dimensional case. We will not include the proof of this claim here, as the upper bound in \eqref{00mupper} is sufficient for our purposes.

\begin{prop}\label{prop:tail_m}
\begin{equation}
\Prob(m_j> u) = 
\begin{cases}
O\big(u^{-2} (\log u)^{-1}\big) \ind{u\leq r_j} & (d=2) \\
O\big(u^{-(1+\frac{d}{2})}\big) \ind{u\leq r_j} & (d\geq 3) .
\end{cases}
\end{equation}
\end{prop}

\begin{proof}
This follows from Propositions \ref{prop:tail2} and \ref{prop:tail3}, since $m_j\leq \mu_j$ and $m_j\leq r_j$.
\end{proof}

\begin{prop}\label{prop:m1bound}
\begin{equation}\label{m1}
\Exp(m_j) =  \xibar + O( r_j^{-1}) .
\end{equation}
\end{prop}

\begin{proof}
We have
\begin{equation}
\begin{split}
\Exp(m_j) & = \frac{1}{v_{d-1}} \int_\UB \int_\UB K_{1,r_j}(\vecw,\vecz)\,d\vecw\,d\vecz \\
& = \int_0^{r_j} x \Psi_0(x)\, dx \\
& = \xibar - \int_{r_j}^\infty x \Psi_0(x)\, dx,
\end{split}
\end{equation}
and \eqref{m1} follows from  the asymptotics \eqref{PHIBARXILARGETHMRES1}. 
\end{proof}

\begin{prop}\label{prop:mbound}
\begin{equation}
\Exp(m_j^2) =  
\begin{cases}
O(\log\log j) & (d=2)\\
O(1) & (d\geq 3).
\end{cases}
\end{equation}
\end{prop}

\begin{proof}
This is a direct corollary of Proposition \ref{prop:tail_m}.
\end{proof}

\begin{prop}\label{prop:aasy}
For $d\geq 2$ and $j\to\infty$,
\begin{equation}\label{eq:aasy}
\Exp(a_j^2) = \frac{\Theta_d}{2}\, \log j +  
O(\log\log j) ,
\end{equation}
\begin{equation}\label{eq:basy}
\Exp(b_j^2) = \frac{\Theta_d}{2}\, \log j +  
O(\log\log j) .
\end{equation}
\end{prop}

\begin{proof}
We have
\begin{equation}
\Exp(b_j^2) = \frac{1}{v_{d-1}} \int_\UB \int_\UB K_{2,r_j}(\vecw,\vecz)\,d\vecw\,d\vecz
= \int_0^{r_j} x^2 \Psi_0(x)\, dx,
\end{equation}
and hence \eqref{eq:basy} follows from the asymptotics \eqref{PHIBARXILARGETHMRES1}. Relation \eqref{eq:aasy} is a consequence of  Proposition \ref{prop:mbound}.
\end{proof}

\begin{prop}\label{prop:amom}
For $d\geq 2$, 
\begin{equation}
\Exp(a_j^{4}) \leq  \Exp(b_j^{4}) = O\big(r_j^{2}\big) .
\end{equation}
\end{prop}

\begin{proof}
We have
\begin{equation}
\begin{split}
\Exp(b_j^{4}) & = \frac{1}{v_{d-1}} \int_\UB \int_\UB  \frac{K_{2,r_j}(\vecw,\vecz)^2}{K_0(\vecw,\vecz)}\,d\vecw\,d\vecz \\
& \leq 2 \int_0^{r_j} x^{4} \Psi_0(x)\, dx ,
\end{split} 
\end{equation}
and the claim follows from \eqref{PHIBARXILARGETHMRES1}.
\end{proof}

\begin{prop}\label{betatail2}
Let $d=2$. For $u\to\infty$,
\begin{equation}\label{aasy}
\Prob(\alpha_n> u) \sim \frac{1}{2\pi^2 u^2} .
\end{equation}
\begin{equation}\label{basy}
\Prob(\beta_n> u) \sim \frac{1}{2\pi^2 u^2} .
\end{equation}
\end{prop}

\begin{proof}
We follow the same strategy as in the proof of Proposition \ref{prop:tail2}. For 
\begin{equation}
f(y)=\frac{6y^2(1+y)^2}{1+2y} \,\ln(1+y^{-1}),
\end{equation}
the variable substitution \eqref{wxyz} yields
\begin{equation}
\begin{split}
\Prob(\beta_n>u)
= &\frac{12}{\pi^2} \int_{0}^\infty \int_{\frac{2}{1+2y}}^{\frac{2}{1+y}} \ind{x^2 f(y)<u^{-2}} \ln(1+y^{-1})\; dx\, dy \\
=&\frac{12}{\pi^2} \int_{0}^\infty \ind{\frac{4 f(y)}{(1+y)^2}<u^{-2}} \left( \frac{2}{1+y}-\frac{2}{1+2y}\right)  \ln(1+y^{-1}) \; dy\\ 
& +\frac{12}{\pi^2} \int_{0}^\infty \ind{\frac{4 f(y)}{(1+2y)^2}<u^{-2}<\frac{4 f(y)}{(1+y)^2}} \left( \frac{1}{uf(y)}-\frac{2}{1+2y}\right)  \ln(1+y^{-1}) \; dy .
\end{split}
\end{equation}
The leading order contribution comes from the first term, which evaluates to
\begin{multline}
\frac{12}{\pi^2} \int_{0}^\infty \ind{\frac{24 y^2}{1+2y} \ln(1+y^{-1})<u^{-2}} \left(2y+O(y^2)\right)  \ln(1+y^{-1}) \; dy\\ 
 = \frac{12}{\pi^2} \left(y^2 \ln(1+y^{-1}) + O(y^2) \right)  \bigg|_{0_+}^{\frac{24 y^2}{1+2y} \ln(1+y^{-1})=u^{-2}} = \frac{1}{2\pi^2 u^2} (1+o(1)).
\end{multline} 
This proves \eqref{basy}. To see that \eqref{aasy} has the same asymptotics, recall that $\beta_n^2-\alpha_n^2=\mu_n^2$ and \eqref{masy}.
\end{proof}

\begin{prop}\label{betatail3}
Let $d\geq 3$. There are constants $c_2> c_1>0$, such that for $u\geq 1$,
\begin{equation}\label{00bupper}
c_1 u^{-2} \leq \Prob(\alpha_n> u)\leq \Prob(\beta_n> u) \leq c_2 u^{-2} .
\end{equation}
\end{prop}

\begin{proof}
We exploit the bounds in Lemma \ref{K2upperlower}. For the upper bound,
\begin{equation}\label{00st122}
\int_{K_2> u^2 K_0} K_0 \, d\vecw\,d\vecz 
\ll \int\limits_{\substack{\max( t^{-\frac{d-2}{2}}, t^{-\frac{d-1}{2}} \varphi) \gg u^2 \\ 0 \leq \varphi\leq \frac{\pi}{2}}} d\vecw\,d\vecz
+  \int\limits_{\substack{\min( t^{-(1+\frac{d}{2})} , t^{-\frac{d+1}{2}+\frac1d} (\pi-\varphi)^{-(1+\frac2d)})  \gg u^2\\ \frac{\pi}{2} \leq \varphi\leq \pi}} d\vecw\,d\vecz  .
\end{equation}
Set
\begin{equation}
\alpha:=\frac{4}{d+2},\qquad \beta:= \frac{4d}{d(d+1)-2}, \qquad \gamma:=\frac{2d+4}{d(d+1)-2},
\end{equation}
and note that
\begin{equation}
\frac{\alpha-\beta}{\gamma}= -\frac{2}{d+2}.
\end{equation}
Using polar co-ordinates as before, the second term in \eqref{00st122} evaluates to
\begin{equation}
\begin{split}
& \int_0^{\pi/2}
\min\big(u^{-\alpha}, u^{-\beta} \phi^{-\gamma} \big)^2 (\sin\phi)^{d-3} d\phi \\
& = u^{-2\alpha} \int_0^{u^{\frac{\alpha-\beta}{\gamma}}} (\sin\phi)^{d-3} d\phi +
u^{-2\beta}  \int_{u^{\frac{\alpha-\beta}{\gamma}}}^{\pi/2} \phi^{-2\gamma}  (\sin\phi)^{d-3} d\phi \\
& \ll u^{-2\alpha}  u^{(d-2) \frac{\alpha-\beta}{\gamma}} + u^{-2\beta}  u^{(d-2-2\gamma)\frac{\alpha-\beta}{\gamma}} = 2 u^{-2}.
\end{split}
\end{equation}
A similar calculation shows that the first term in \eqref{00st122} produces a lower order contribution.
This establishes the upper bound in \eqref{00bupper}. For the lower bound for $\Prob(\beta_n> u)$,
\begin{equation}\label{00st1122}
\begin{split}
\int_{K_2> u^2 K_0} K_0 \, d\vecw\,d\vecz 
& \gg \int\limits_{\substack{\min( t^{-(1+\frac{d}{2})} , (\pi-\varphi)^{-(d+2)} ) \gg u^2\\ \frac{\pi}{2} \leq \varphi\leq \pi}} d\vecw\,d\vecz  \\
& \gg u^{-2\alpha} \int_0^{u^{-\frac{2}{d+2}}} (\sin\phi)^{d-3} d\phi \\ 
& \gg u^{-2\alpha} u^{-2\frac{d-2}{d+2}} = u^{-2} .
\end{split}
\end{equation}
The lower bound for $\Prob(\alpha_n> u)$ follows by combining the lower bound for $\Prob(\beta_n> u)$ with \eqref{00mupper}. 
\end{proof}

\section{Spectral gaps}\label{sec:spectral}

Let $V$ be a finite-dimensional real vector space with inner product $\langle\,\cdot\, , \,\cdot\,\rangle_V$ and norm $\| x \|:=\langle x,x\rangle^{1/2}$.
Denote by $\scrH=\L^2(\UB,V, v_{d-1}^{-1} d\vecw)$ the Hilbert space of square-integrable functions $\UB\to V$ with inner product 
\begin{equation}
\langle f,g \rangle := \frac{1}{ v_{d-1}} \int_\UB \langle f(\vecw), g(\vecw)\rangle_V \, d\vecw,
\end{equation}
and norm $\| f \| := \langle f,f \rangle^{1/2}$. We will also denote by $\|\,\cdot\,\|$ the corresponding operator norm on $\scrH\to\scrH$. In the following, $(\rho,V)$ will denote a representation of $\SO(d)$ with group homomorphism $\rho: \SO(d)\to \O(V)$. 

Define the following operators on $\scrH$:
\begin{equation}
Pf(\vecw):= \int_\UB K_0(\vecw,\vecz) f(\vecz) \,d\vecz,
\end{equation}
\begin{equation}
\Pi f(\vecw):= \frac{1}{ v_{d-1}} \int_\UB  f(\vecz) \,d\vecz,
\end{equation}
\begin{equation}
Uf(\vecw):= \rho(S(\vecw)) f(\vecw)  .
\end{equation}
We have
\begin{equation}\label{commute}
\Pi P = P \Pi = \Pi .
\end{equation}
Denote by $\scrH_0:=\Pi\scrH$ the subspace of constant functions, and by $\scrH_1=(I-\Pi)\scrH$ its orthogonal complement. (This means that all components of $f\in\scrH_1$ have zero mean.) Note that for $f\in\scrH_0$ we have $Pf=f$, and for  $f\in\scrH_1$ we have $Pf\in\scrH_1$.

\begin{prop}\label{prop:Pgap}
The operator $P$ has the spectral gap $1-\omega_0$ with
\begin{equation}\label{nu0}
\omega_0 := \| P-\Pi\| \leq 1-\frac{1}{2^d\zeta(d)}.
\end{equation}
\end{prop}

\begin{proof}
This follows from the standard Doeblin argument. Note that, since $K_0(\vecw,\vecz)$ is the kernel of a stochastic transition operator with respect to $d\vecw$ on $\UB$, we have 
\begin{equation}
J:= v_{d-1} \inf K_0(\vecw,\vecz)\leq 1. 
\end{equation}
If $J=1$, we have $P=\Pi$ and thus $\omega_0=0$. Assume therefore $0\leq J<1$.
Then
\begin{equation}
Q:= (1-J)^{-1} ( P - J \Pi)
\end{equation}
is itself a stochastic transition operator, with the same stationary measure. Using \eqref{commute}, we can write
\begin{equation}
P=\Pi + (1-J) (I-\Pi)Q(I-\Pi) ,
\end{equation}
and so
\begin{equation}
\|P-\Pi\| \leq (1-J) \| Q \| = 1-J .
\end{equation}
The claim of the proposition now follows from \eqref{K0lower}.
\end{proof}

\begin{lem}\label{repgap}
Let $\theta:[0,1)\to \RR$ be measurable, so that
\begin{equation}\label{irr}
\meas\{ w\in[0,1) : \theta(w) \notin \pi \QQ \} >0 ,
\end{equation}
and let $(\rho,V)$ be a non-trivial irreducible representation of $\SO(d)$. Then
\begin{equation}
\delta_\rho:=\bigg\| \frac{1}{ v_{d-1}} \int_\UB \rho(S(\vecw)) \, d\vecw \bigg\| < 1.
\end{equation}
\end{lem}

\begin{proof}
For any fixed $\vece\in\US$, let $\Lambda_\vece$ be the push-forward of Lebesgue measure on $[0,1)$ under the map
\begin{equation}
[0,1)\to \SO(d), \qquad w\mapsto S(w\vece) .
\end{equation}
The group generated by the support of $\Lambda_\vece$ is, by assumption \eqref{irr}, dense in 
the subgroup 
\begin{equation}
\big\{ E(\phi\vece) : \phi\in[0,2\pi) \big\} \simeq \SO(2) ,
\end{equation}
with $E(\vecx)$ as in \eqref{SbLG222}.
Next, let $\Lambda$ be the push-forward of Lebesgue measure on $\UB$ under the map 
\begin{equation}
\UB\to \SO(d), \qquad \vecw\mapsto S(\vecw). 
\end{equation}
The above observation, together with the fact that 
\begin{equation}
\big\{E(\vecx) : \vecx\in\scrB_{2\pi}^{d-1} \big\} 
\end{equation}
generates $\SO(d)$, implies that the group generated by the support of $\Lambda$ is dense in $\SO(d)$.
The claim now follows from well known arguments \cite{Stromberg:1960vj}.
\end{proof}

\begin{prop}\label{prop2}
Let $\theta$ and $(\rho,V)$ be as in Lemma \ref{repgap}. Then the operator $PU$ has spectral radius
\begin{equation}\label{srad}
\omega_\rho := \lim_{n\to\infty} \| (PU)^n \|^{1/n} <1.
\end{equation}
\end{prop}

\begin{proof}
It is sufficient to prove $\|PUP\|<1$, that is
\begin{equation}\label{sup1}
\sup_{f\in\scrH : f\neq 0} \frac{\|PUPf \|^2}{\|f\|^2} < 1.
\end{equation}
We may restrict to functions of the form $f=\alpha f_0+f_1$, where $\alpha>0$ and $f_0\in\scrH_0$, $f_1\in\scrH_1$ with $\|f_0\|=\|f_1\|=1$. Note that in this case $\|f\|^2=\alpha^2+1$, and hence the supremum \eqref{sup1} equals
\begin{equation}\label{sup2}
\sup_{\alpha>0} \frac{1}{\alpha^2+1} \sup_{f_0,f_1} \|PUP(\alpha f_0+f_1) \|^2 .
\end{equation}
Now,
\begin{equation}
\begin{split}
\sup_{f_1} \| PUP(\alpha f_0+f_1) \|^2 & = \sup_{f_1} \| PU(\alpha f_0+Pf_1)\|^2 \\
& \leq \sup_{f_1} \| PU(\alpha f_0+\omega_0 f_1) \|^2
\end{split}
\end{equation}
since 
\begin{equation}
\{ Pf_1 : f_1 \in\scrH_1 , \; \| f_1 \| =1 \} \subseteq \{ f_1 \in\scrH_1 , \; \| f_1 \| \leq \omega_0 \} .
\end{equation}
We have
\begin{equation}
U(\alpha f_0+\omega_0 f_1) = \alpha y_0 + \omega_0 y_1 + \alpha \tilde y_0 + \omega_0 \tilde y_1  ,
\end{equation}
where 
\begin{equation}
y_0 : = \Pi U f_0 \in\scrH_0 , \qquad \tilde y_0 := (I-\Pi) U f_0  \in\scrH_1,
\end{equation}
\begin{equation}
y_1 := \Pi U f_1 \in\scrH_0, \qquad \tilde y_1 :=(I- \Pi) U f_1 \in\scrH_1.
\end{equation}
Therefore,
\begin{equation}\label{22nd}
\begin{split}
& \| PU(\alpha f_0+\omega_0 f_1) \|^2 \\
& =  
\| \alpha y_0 + \omega_0 y_1\|^2 + \| P( \alpha \tilde y_0 + \omega_0 \tilde y_1) \|^2 \\
& \leq   
\| \alpha y_0 + \omega_0 y_1\|^2 + \omega_0^2 \| \alpha \tilde y_0 + \omega_0 \tilde y_1 \|^2 \\
& = \alpha^2 ( \omega_0^2 + (1-\omega_0^2) \| y_0 \|^2) + 2\alpha\omega_0 (1-\omega_0^2) \langle y_0,y_1\rangle + \omega_0^2 (\omega_0^2+ (1-\omega_0^2) \|y_1\|^2 ) .
\end{split}
\end{equation}
In the last equality we have used the relations
\begin{equation}\label{tilde1}
\langle \tilde y_0, \tilde y_1 \rangle = - \langle  y_0,y_1 \rangle ,
\end{equation}
which follows from $\langle Uf_0,Uf_1 \rangle =\langle f_0,f_1 \rangle=0$,
and
\begin{equation}\label{tilde2}
\| \tilde y_0 \|^2 = 1-\| y_0 \|^2 , \qquad \| \tilde y_1 \|^2 = 1-\| y_1 \|^2  .
\end{equation}
Since $f_0$ is a constant function with $\| f_0 \|=1$, we have by Lemma \ref{repgap} $\|y_0\|\leq \delta_\rho<1$. Furthermore $\|y_1\|\leq 1$. This shows
\begin{equation}\label{sup3}
\sup_{f\in\scrH : f\neq 0} \frac{\|PUPf \|^2}{\|f\|^2} 
\leq \sup_{\alpha>0}\; \sup_{\substack{\|y_0\|\leq \delta_\rho\\ \|y_1\|\leq 1}} B(\alpha),
\end{equation}
with
\begin{equation}
B(\alpha) := \frac{\alpha^2 ( \omega_0^2 + (1-\omega_0^2) \| y_0 \|^2) + 2\alpha\omega_0 (1-\omega_0^2) \langle y_0,y_1\rangle + \omega_0^2 (\omega_0^2+(1-\omega_0^2) \|y_1\|^2)}{\alpha^2+1} .
\end{equation}
The final step in the proof of Proposition \ref{prop2} is now to show that 
\begin{equation}\label{sup4}
\sup_{\alpha>0} \sup_{\substack{\|y_0\|\leq \delta_\rho\\ \|y_1\|\leq 1}} B(\alpha) < 1.
\end{equation}
To achieve this, first note that
\begin{equation}
\sup_{\substack{\|y_0\|\leq \delta_\rho\\ \|y_1\|\leq 1}} B(0) =\sup_{\substack{\|y_1\|\leq 1}}  \omega_0^2 (\omega_0^2+(1-\omega_0^2) \|y_1\|^2 )\leq \omega_0^2 < 1 
\end{equation}
and
\begin{equation}
\lim_{\alpha\to\infty} \sup_{\substack{\|y_0\|\leq \delta_\rho\\ \|y_1\|\leq 1}} B(\alpha) = \sup_{\substack{\|y_0\|\leq \delta_\rho}} \omega_0^2 + (1-\omega_0^2) \| y_0 \|^2 < 1.
\end{equation}
To prove, \eqref{sup4}, it is therfore sufficient that the quadratic equation 
\begin{equation}\label{qeq}
B(\alpha)=1
\end{equation}
has no positive real solution for all $\|y_0\|\leq \delta_\rho$, $\|y_1\|\leq 1$.
This in turn holds, if the discriminant of Eq.~\eqref{qeq} is strictly negative, i.e.
\begin{equation}
\sup_{\substack{\|y_0\|\leq \delta_\rho\\ \|y_1\|\leq 1}} \big[ - 4 (1-\omega_0^2)^2 \big\{  1-\|y_0\|^2 + \omega_0^2 \big[ (1 - \|y_0\|^2)(1- \|y_1\|^2)  -   \langle y_0,y_1\rangle^2 \big] \big\}  \big] < 0.
\end{equation}
Because $1-\omega_0^2>0$ and $1-\|y_0\|^2\geq 1-\delta_\rho^2>0$, it remains to be shown that
\begin{equation}\label{finale}
\langle y_0,y_1\rangle^2\leq (1 - \|y_0\|^2)(1- \|y_1\|^2)  .
\end{equation}
To this end, apply eqs.~\eqref{tilde1}, \eqref{tilde2} and the Cauchy-Schwarz inequality,
\begin{equation}
|\langle y_0,y_1\rangle|=|\langle \tilde y_0,\tilde y_1\rangle|\leq \|\tilde y_0\| \|\tilde y_1\| = \sqrt{(1 - \|y_0\|^2)}\sqrt{(1- \|y_1\|^2)}  ,
\end{equation}
which completes the proof.
\end{proof}

\section{Exponential mixing}\label{sec:exponential}

We will now apply the spectral estimates of the previous section to obtain exponential mixing rates.

Denote by $(\rho_1,V_1)=(\id,\RR^d)$ the natural representation of $\SO(d)$, and by $(\rho_2,V_2)$ the adjoint representation of $\SO(d)$ on the vector space $V_2$ of real symmetric traceless $d\times d$ matrices defined by
\begin{equation}
\rho_2(R) : M \mapsto R M \trans R .
\end{equation}
The inner product on $V_1$ is the standard Euclidean inner product
\begin{equation}
\langle \vecx_1,\vecx_2 \rangle_{V_1} := \vecx_1\cdot \vecx_2,
\end{equation}
and on $V_2$ the Hilbert-Schmidt inner product
\begin{equation}
\langle M_1,M_2 \rangle_{V_2} := \tr( M_1 M_2).
\end{equation}

\begin{prop}\label{prop:exp_mix}
Fix any $\omega\in[\omega_0,1)\cap (\omega_{\rho_1},1)\cap(\omega_{\rho_2},1)$, $m\in\NN$ and $p=0,1,2$. Then there is a constant $C_m>0$ such that, for all $n_1,n_2\in\NN$, $\vecv_0,\vece\in\US$ and all measurable $f,g: (\UB)^{m+1} \to \RR$ with
\begin{equation}
\Exp\big( f(\veceta_0,\ldots,\veceta_m)^2 \big)<\infty, \qquad \Exp\big( g(\veceta_0,\ldots,\veceta_m)^2 \big)<\infty ,
\end{equation}
we have
\begin{multline}\label{zone0}
\big| \Cov\big( (\vece\cdot\vecV_{n_1})^p f(\veceta_{n_1},\ldots,\veceta_{n_1+m}), (\vece\cdot\vecV_{n_2})^p g(\veceta_{n_2},\ldots,\veceta_{n_2+m})\big)\big| \\ \leq C_m\, \omega^{|n_1-n_2|-m} 
\sqrt{\Exp\big( f(\veceta_0,\ldots,\veceta_m)^2 \big)}\sqrt{\Exp\big( g(\veceta_0,\ldots,\veceta_m)^2 \big)}.
\end{multline}
\end{prop}

(Note that we have fixed $\vecV_0=\vecv_0$, which breaks stationarity.)

\begin{proof}
Assume without loss of generality that $n_1\leq n_2$. We have
\begin{multline}
\Cov\big( (\vece\cdot\vecV_{n_1})^p f(\veceta_{n_1},\ldots,\veceta_{n_1+m}), (\vece\cdot\vecV_{n_2})^p g(\veceta_{n_2},\ldots,\veceta_{n_2+m})\big) \\
= \Exp \big[ \Cov\big( (\vece\cdot\vecV_{n_1})^p f(\veceta_{n_1},\ldots,\veceta_{n_1+m}), (\vece\cdot\vecV_{n_2})^p g(\veceta_{n_2},\ldots,\veceta_{n_2+m})\mid \vecV_{n_1} \big) \big] .
\end{multline}
It is therefore sufficient to prove \eqref{zone0} conditioned on $\vecV_{n_1}$ with a constant $C_m$ independent on $\vecV_{n_1}$, or equivalently, to show that form any $n\in\ZZ_{\geq 0}$, $\vecv_0,\vece\in\US$,  
\begin{multline}\label{zone}
\big| \Cov\big( f(\veceta_{0},\ldots,\veceta_{m}), (\vece\cdot\vecV_{n})^p g(\veceta_{n},\ldots,\veceta_{n+m})\big)\big| \\ \leq C_m\, \omega^{n-m} 
\sqrt{\Exp\big( f(\veceta_0,\ldots,\veceta_m)^2 \big)}\sqrt{\Exp\big( g(\veceta_0,\ldots,\veceta_m)^2 \big)}.
\end{multline}
The case $n\leq m$ follows from the Cauchy-Schwarz inequality. We assume therefore in the following that $n>m$.

{\em Case A: $p=0$.} Define the functions $\UB\to\RR$
\begin{equation}
\tilde f(\vecw) := \Exp\big( f(\veceta_0,\ldots,\veceta_m) \bigm| \veceta_m=\vecw\big),
\end{equation}
\begin{equation}
\tilde g(\vecw) := \Exp\big( g(\veceta_0,\ldots,\veceta_m) \bigm| \veceta_0=\vecw\big) .
\end{equation}
Then
\begin{equation}
\begin{split}
\Cov\big( f(\veceta_0,\ldots,\veceta_m), g(\veceta_n,\ldots,\veceta_{n+m})\big) & =
\Cov\big( \tilde f(\veceta_m), \tilde g(\veceta_n)\big) \\ & =
\langle (I-\Pi) \tilde f, (P^{n-m}-\Pi) \tilde g \rangle \\ & = \langle (I-\Pi) \tilde f, (P-\Pi)^{n-m} \tilde g \rangle,
\end{split}
\end{equation}
and so
\begin{equation}
\begin{split}
\big| \Cov\big( f(\veceta_0,\ldots,\veceta_m), g(\veceta_n,\ldots,\veceta_{n+m})\big) \big| 
& \leq \| \tilde f \| \| (P-\Pi)^{n-m} \tilde g \| \\
& \leq \omega_0^{n-m} \| \tilde f \| \| \tilde g \| 
\end{split}
\end{equation}
in view of Proposition \ref{prop:Pgap}. Finally,
\begin{equation}
\| \tilde f \|^2 \leq \Exp\big( f(\veceta_0,\ldots,\veceta_m)^2 \big),
\end{equation}
since $\tilde f$ is obtained from $f$ via orthogonal projection, thus decreasing the $\L^2$-norm  (likewise for $\tilde g$). This proves \eqref{zone} for $p=0$.

{\em Case B: $p=1$.} Set $\tilde\vece= R(\vecv_0)^{-1} \vece$, and
\begin{equation}
 \overline f:=\Exp f(\veceta_0,\ldots,\veceta_m) .
\end{equation}
Now
\begin{equation}
\begin{split}
& \Cov\big(  f(\veceta_0,\ldots,\veceta_m), (\vece\cdot\vecV_{n}) g(\veceta_n,\ldots,\veceta_{n+m})\big) \\ 
& = \Cov\big( f(\veceta_0,\ldots,\veceta_m), (\tilde\vece\cdot S(\veceta_1)\cdots S(\veceta_{n})\vece_1) g(\veceta_n,\ldots,\veceta_{n+m})\big) \\
& =  \Exp\big( [f(\veceta_0,\ldots,\veceta_m)-\overline f] (\tilde\vece\cdot S(\veceta_1)\cdots S(\veceta_{n})\vece_1) g(\veceta_n,\ldots,\veceta_{n+m})\big) .
\end{split}
\end{equation}
By using the vector-valued functions $\UB\to\RR^d$
\begin{equation}
\tilde f(\vecw) := \Exp\big( [ f(\veceta_0,\ldots,\veceta_m) -\overline f ] \trans S(\veceta_m)\cdots \trans S(\veceta_1) \bigm| \veceta_m=\vecw\big) \, \tilde\vece,
\end{equation}
\begin{equation}
\tilde g(\vecw) := \Exp\big( g(\veceta_0,\ldots,\veceta_m) \bigm| \veceta_0=\vecw\big) \, \vece_1, 
\end{equation}
we find
\begin{equation}
\begin{split}
& \Cov\big(  f(\veceta_0,\ldots,\veceta_m), (\vece\cdot\vecV_{n}) g(\veceta_n,\ldots,\veceta_{n+m})\big) \\ 
& =  \Exp\big( \tilde f(\veceta_m) \cdot S(\veceta_{m+1})\cdots S(\veceta_{n}) \tilde g(\veceta_n) \big) \\
& =  \langle \tilde f , (PU)^{n-m} \tilde g \rangle .
\end{split}
\end{equation}
We conclude from Proposition \ref{prop2} applied to the natural representation $(\rho_1,V_1)$:
\begin{equation}
\big|  \Cov\big( f(\veceta_0,\ldots,\veceta_m), (\vece\cdot\vecV_{n}) g(\veceta_n,\ldots,\veceta_{n+m})\big) \big| \leq c\, \omega^{n-m} \| \tilde f \| \| \tilde g \| 
\end{equation}
for some $c>0$.
Finally, because orthogonal projection decreases the $\L^2$-norm, 
\begin{equation}
\begin{split}
 \| \tilde f \|^2 & \leq \Exp\big( \| [ f(\veceta_0,\ldots,\veceta_m) -\overline f ] \trans S(\veceta_m)\cdots \trans S(\veceta_1) \tilde\vece \|_{V_1}^2 \big)  \\
 & =  \Exp\big( [ f(\veceta_0,\ldots,\veceta_m) -\overline f ]^2 \big) \\
 & \leq \Exp\big( f(\veceta_0,\ldots,\veceta_m)^2 \big). 
 \end{split}
\end{equation}

{\em Case C: $p=2$.}
The vector $\tilde\vece$ and the expectation $\overline f$ are defined as in Case B. We have
\begin{equation}\label{varav}
\begin{split}
& \Cov\big(  f(\veceta_0,\ldots,\veceta_m), (\vece\cdot\vecV_{n})^2 g(\veceta_n,\ldots,\veceta_{n+m})\big) \\ 
& = \Cov\big(  f(\veceta_0,\ldots,\veceta_m), (\tilde\vece\cdot S(\veceta_1)\cdots S(\veceta_n)\vece_1)^2 g(\veceta_n,\ldots,\veceta_{n+m})\big) \\
& =  \Exp\big( [f(\veceta_0,\ldots,\veceta_m)-\overline f] (\tilde\vece\cdot S(\veceta_1)\cdots S(\veceta_n)\vece_1)^2 g(\veceta_n,\ldots,\veceta_{n+m})\big) .
\end{split}
\end{equation}
We write
\begin{equation}\label{symsq}
\begin{split}
(\tilde\vece\cdot S(\veceta_1)\cdots S(\veceta_n)\vece_1)^2 
%& =   \tr \big[ \tilde\vece\trans\tilde\vece S(\veceta_1)\cdots S(\veceta_n) \vece_1\trans\vece_1 \trans(S(\veceta_1)\cdots S(\veceta_n)) \big] \\
& =   \frac1d+\tr \big[ \tilde E S(\veceta_1)\cdots S(\veceta_n) E_1 \trans(S(\veceta_1)\cdots S(\veceta_n)) \big] \\
& =   \frac1d+ \big\langle \tilde E ,\rho\big(S(\veceta_1)\cdots S(\veceta_n)\big) E_1 \big\rangle_{V_2} ,
\end{split}
\end{equation}
where $V_2$ is the vector space of symmetric traceless $d\times d$ matrices, and
\begin{equation}
E_1:=\vece_1\otimes\vece_1-\frac1d I_d \in V_2, \qquad \tilde E:=\tilde\vece\otimes\tilde\vece-\frac1d I_d \in V_2.
\end{equation}
The constant term $\frac1d$ in \eqref{symsq} contributes to \eqref{varav} the term 
\begin{multline}
\frac1d   \big| \Exp\big( [f(\veceta_0,\ldots,\veceta_m)-\overline f]  g(\veceta_n,\ldots,\veceta_{n+m})\big) \big| \\ \leq \frac1d \; \omega_0^{n-m} 
\sqrt{\Exp\big( f(\veceta_0,\ldots,\veceta_m)^2 \big)}\sqrt{\Exp\big( g(\veceta_0,\ldots,\veceta_m)^2 \big)}.
\end{multline}
This follows from our discussion in Case A ($p=0$). 
The non-constant term in \eqref{symsq} is handled in analogy with Case B. 
Define functions $\UB\to V_2$
\begin{equation}
\tilde f(\vecw) := \Exp\big( [ f(\veceta_0,\ldots,\veceta_m) -\overline f ] \rho(\trans S(\veceta_m)\cdots \trans S(\veceta_1)) \bigm| \veceta_m=\vecw\big) \, \tilde E,
\end{equation}
\begin{equation}
\tilde g(\vecw) := \Exp\big( g(\veceta_0,\ldots,\veceta_m) \bigm| \veceta_0=\vecw\big) \, E_1, 
\end{equation}
so that the non-constant contribution to \eqref{varav} becomes
\begin{equation}
\begin{split}
&  \Exp\big( [f(\veceta_0,\ldots,\veceta_m)-\overline f] (\tilde\vece\cdot S(\veceta_1)\cdots S(\veceta_n)\vece_1)^2 g(\veceta_n,\ldots,\veceta_{n+m})\big) \\
& =  \Exp\big( [f(\veceta_0,\ldots,\veceta_m)-\overline f] \big\langle \tilde E ,\rho\big(S(\veceta_1)\cdots S(\veceta_n)\big) E_1 \big\rangle_{V_2} \,g(\veceta_n,\ldots,\veceta_{n+m})\big) \\
& =  \langle \tilde f , (PU)^{n-m} \tilde g \rangle .
\end{split}
\end{equation}
We now apply Proposition \ref{prop2} with the adjoint representation $(\rho_2,V_2)$. This yields the desired bound.
\end{proof}

We will also require the following estimate. 

\begin{prop}\label{prop:exp_exp}
Fix $\omega$ and $m$ as in Proposition \ref{prop:exp_mix}. Then there is a constant $\widetilde C>0$ such that, for all $n\in\NN$, $\vecv_0,\vece\in\US$ and all measurable $g: (\UB)^{m+1} \to \RR$ with
\begin{equation}
\Exp\big( g(\veceta_0,\ldots,\veceta_m)^2 \big)<\infty, 
\end{equation}
we have
\begin{equation}\label{eq:exp}
\big| \Exp\big( (\vece\cdot\vecV_{n})  g(\veceta_{n},\ldots,\veceta_{n+m}) \big) \big| 
\leq \widetilde C\, \omega^n \sqrt{\Exp\big( g(\veceta_0,\ldots,\veceta_m)^2 \big)} .
\end{equation}
\begin{equation}\label{eq:exp2}
\big| \Exp\big( [ (\vece\cdot\vecV_{n})^2- d^{-1}] g(\veceta_{n},\ldots,\veceta_{n+m}) \big) \big| 
\leq \widetilde C\, \omega^n \sqrt{\Exp\big( g(\veceta_0,\ldots,\veceta_m)^2 \big)} .
\end{equation}
\end{prop}

\begin{proof}
With $\tilde g$ defined as in the previous proof, Case B, we have
\begin{equation}
\Exp\big( (\vece\cdot\vecV_{n})  g(\veceta_{n},\ldots,\veceta_{n+m}) \big)  
= \langle 1 , (PU)^n \tilde g \rangle.
\end{equation}
The bound \eqref{eq:exp} now follows from Proposition \ref{prop2}. Relation \eqref{eq:exp2} follows similarly from Case C of the previous proof.\end{proof}

\section{Proof of the main lemmas}\label{sec:proof}

We now turn to the proofs of the four main lemmas in Section \ref{sec:outline}. 

\begin{proof}[Proof of Lemma \ref{mainlem0}]
We have
\begin{equation}
\| \vecQ_n-\vecQ_n' \| \leq  \sum_{j=1}^n \zeta_j  
\end{equation}
with 
\begin{equation}
\zeta_j:=\xi_j \ind{\xi_j^2 >   j (\log j)^\gamma} .
\end{equation}
By \eqref{PHIBARXILARGETHMRES1}, we have
\begin{equation}
\prob{\zeta_j\not=0} = O\big(j^{-1} (\log j)^{-\gamma}\big)
\end{equation}
This is summable (since $\gamma>1$) and so, by the Borel-Cantelli lemma, $\zeta_j\not=0$ only for finitely many $j$. This proves Lemma \ref{mainlem0}.
\end{proof}

\begin{proof}[Proof of Lemma \ref{mainlem1}] 
Set $\zeta_j:=(\vece\cdot  \vecV_{j-1})\, m_j$. We need to show that, for every $\vece\in\US$, the sequence of random variables 
\begin{equation}\label{eR}
\begin{cases}
\frac{\sum_{j=1}^n  \zeta_j}{\sqrt{n \log\log n}} & (d=2)\\[10pt]
\frac{\sum_{j=1}^n  \zeta_j}{\sqrt{n}} & (d\geq 3)
\end{cases}
\end{equation}
is tight. Now choose in Proposition \ref{prop:exp_exp} $m=1$ and $g(\vecw,\vecz)=K_{1,r_j}(\vecw,\vecz)/K_0(\vecw,\vecz)$, and use Proposition \ref{prop:mbound} to bound $\Exp(g(\veceta_0,\veceta_1)^2)=\Exp(m_j^2)=O(\log\log j)$ for $d=2$, and $=O(1)$ for $d\geq 3$. This proves $\Exp(\zeta_j) = O( \omega^j)$. Therefore
\begin{equation}
\var{\zeta_j^2} = \expect{\zeta_j^2} + O(\omega^{2j}) .
\end{equation}
Proposition \ref{prop:mbound} yields
\begin{equation}
\expect{\zeta_j^2}  \leq \expect{ m_j^2 }\\
=  
\begin{cases}
O(\log\log j) &(d=2)\\
O(1) & (d\geq 3).
\end{cases}
\end{equation}
Due to Proposition \ref{prop:exp_mix}, we also have
\begin{equation}
\cov{\zeta_i} {\zeta_j} \le \sqrt{\expect{\zeta_i^2}}\sqrt{\expect{\zeta_j^2}} \omega^{|i-j|}.
\end{equation}
Hence
\begin{equation}
\label{sumXtight}
\Exp\big[\big( \sum_{j=1}^n \zeta_j\big)^2\big]
= 
\begin{cases}
O(n\log\log n)  &(d=2)\\
O(n) & (d\geq 3),
\end{cases}
\end{equation}
which establishes the tightness of \eqref{eR} and thus Lemma \ref{mainlem1}.
\end{proof}

\begin{proof}[Proof of Lemma \ref{mainlem2}] 
Let $\zeta_j:= (\vece\cdot \vecV_{j-1})^2 a_j^2$. It is sufficient to prove that for any  unit vector $\vece\in\US$
\begin{equation}\label{eq:mainlem2b2}
\frac{\sum_{j=1}^n  \zeta_j }{n\log n}\toprob  \sigma_d^2  .
\end{equation}
Using \eqref{eq:exp2} in Proposition \ref{prop:exp_exp}, and Proposition \ref{prop:aasy},
\begin{equation}
\expect{\zeta_j}
= \frac{\Theta_d}{2d}\, \log j +
O(\log\log j) ,
\end{equation}
and, by Proposition \ref{prop:amom},
\begin{equation}
\expect{a_j^4}
= O\big(j(\log j)^\gamma\big).
\end{equation}
Furthermore, due to Proposition \ref{prop:exp_mix} $(p=2)$, we have
\begin{equation}
\cov{\zeta_i}{\zeta_j}  \le
\sqrt{\expect{a_i^4}}
\sqrt{\expect{a_j^4}} \; \omega^{|i-j|}  .
\end{equation}
Hence,
\begin{equation}
\expect{\sum_{j=1}^n \zeta_j} = \frac{\Theta_d}{2d}\, n\log n +O(n \log\log n)
\end{equation}
and
\begin{equation}
\var{\sum_{j=1}^n \zeta_j} = O(n^2 (\log n)^\gamma) = o((n\log n)^2)
\end{equation}
since $\gamma<2$.
This proves Lemma \ref{mainlem2}.
\end{proof}

\begin{proof}[Proof of Lemma \ref{mainlem3}] 
In view of the asymptotic relation for $A_n$ in \eqref{eq:mainlem2b} we have to prove that for any $\vareps>0$
\begin{equation}\label{2be}
\frac{\sum_{j=1}^n\condexpect{\tilde \xi_j^2 \ind{\tilde \xi_j^2 > \vareps^2 n\log n}}{\uveceta}}{n\log n}
\toprob 0.
\end{equation}
The {\em lower tail} $\tilde \xi_j < - \vareps \sqrt{n\log n}$ is estimated by
\begin{equation}
\begin{split}
\condexpect{(\xi_j'-m_j)^2\ind{\xi_j'-m_j < -\vareps\sqrt{n\log n} }}{\uveceta}
& \le m_j^2 \ind{m_j>\vareps\sqrt{n\log n}} \\
& \le m_j^2 \ind{m_j>\vareps\sqrt{j\log j}}.
\end{split}
\end{equation}
Proposition \ref{prop:tail_m} yields
\begin{equation}
\prob{m_j>\vareps\sqrt{j\log j}}
=
\begin{cases}
\displaystyle
O((\vareps^2 j (\log j)^2)^{-1}) & (d=2) \\
O((\epsilon^2 j\log j)^{-(\frac12+\frac{d}{4})})  & (d\geq 3) .
\end{cases}
\end{equation}
Since this is summable, we have, by the Borel-Cantelli lemma,
\begin{equation}
\label{lowertail}
\frac{\sum_{j=1}^n
\condexpect{\txi_j^2\ind{\txi_j<-\vareps\sqrt{n\log n} }}{\uveceta}}{n\log n}
\toas
0.
\end{equation}

For  the {\em upper tail} $\tilde \xi_j > \vareps \sqrt{n\log n}$ we have 
\begin{equation}
\begin{split}
\zeta_{n,j} & :=\condexpect{(\xi_j'-m_j)^2\ind{\xi_j-m_j > \vareps\sqrt{n\log n} }}{\uveceta} \\
& \le \condexpect{\xi_j'^2\ind{\xi_j'>\vareps\sqrt{n\log n} }}{\uveceta} \\
& =\condexpect{\xi_j^2\ind{\vareps^2 n\log n <\xi_j^2 \leq j(\log j)^\gamma}}{\uveceta} .
\end{split}
\end{equation}
On the other hand, in view of \eqref{PHIBARXILARGETHMRES1}, we have for $n\to\infty$,
\begin{equation}
\expect{\zeta_{n,j}} \ll \log \sqrt{\frac{ n(\log n)^\gamma}{\vareps^2 n\log n}}
\sim
\frac{\gamma-1}{2} \log \log n ,
\end{equation}
and therefore
\begin{equation}
\label{uppertail}
\frac{\expect{\sum_{j=1}^n \zeta_{n,j}}}{n\log n}
\to0.
\end{equation}
From \eqref{lowertail} and \eqref{uppertail}, the assertion of Lemma \ref{mainlem3} follows.
\end{proof}

\section{General initial data}\label{sec:general}

Up to now we have assumed that $(\xi_1,\veceta_1)$ has density $\Psi_0(x,\vecz)$. We now extend the above results to more general initial data $(\xi_1,\veceta_1)$, where the only assumption is that the marginal distribution of $\veceta_1$ is absolutely continuous with respect to Lebesgue measure on $\UB$.

\begin{proof}[Proof of Theorem \ref{thm:main3} (ii) for general initial data]
Since
\begin{equation}
\frac{\xi_1\vecV_0}{\sqrt{n\log n}} \toprob \vecnull,
\end{equation}
it is sufficient to show that 
\begin{equation}
\frac{\sum_{j=2}^n \xi_j \vecV_{j-1}}{\sigma_d \sqrt{n\log n}}  \Rightarrow \scrN(0,I_d)
\end{equation}
where $\veceta_1$ has (by assumption) an absolutely continuous distribution and $\xi_1=0$. By an obvious re-labelling, this is equivalent to showing that 
\begin{equation}
\frac{\sum_{j=1}^n \xi_j \vecV_{j-1}}{\sigma_d \sqrt{n\log n}}  \Rightarrow \scrN(0,I_d)
\end{equation}
where $\veceta_0$ has an absolutely continuous distribution. In view of the remarks following Eq.~\eqref{eli}, the only difference from the proof of Theorem \ref{thm:main3} is now that $\veceta_0$ is distributed according to an absolutely continuous probability measure, rather than Lebesgue measure. Because {\em tightness}, {\em almost sure convergence} and {\em convergence in probability} continue to hold when passing from a measure to a measure which is absolutely continuous with respect to the first, the Lemmas in Section \ref{sec:outline} remain valid also in the present setting. The proof of Theorem \ref{thm:main3} for general initial data therefore follows from these lemmas in the same way as for the density $\Psi_0(x,\vecz)$, as described at the end of Section \ref{sec:outline}.
\end{proof}

The following proposition shows that, if $(\xi_1,\veceta_1)$ has density $\Psi(x,\vecz)$ (which appears in the continuous-time setting of the Boltzmann-Grad limit, Theorem \ref{thm:MS} (i)), then the marginal distribution of $\veceta_1$ is absolutely continuous. Let
\begin{equation}
\overline\Psi(\vecz):=\int_0^\infty \Psi(x,\vecz)\,dx   
\end{equation}

\begin{prop}\label{prop:main3}
\begin{equation}
\overline\Psi\in\L^1(\UB,d\vecz) .
\end{equation}
\end{prop}

\begin{proof}
The function $\Psi(x,\vecz)$ is continuous, and in view of Corollary 1.2 and Theorem 1.11 in \cite{Marklof:2011di}, uniformly bounded. The latter theorem produces a precise asymptotics of $\Psi(x,\vecz)$, which implies
\begin{equation}
\Psi(x,\vecz)= O(x^{-2+\frac2d}),
\end{equation}
uniformly for all $x>0$, $\vecz\in\UB$.
Thus, for $d\geq 3$, $\overline\Psi(\vecz)$
is uniformly bounded, and hence $\overline\Psi\in\L^1(\UB,d\vecz)$ as required.

In dimension $d=2$, there is an explicit formula for $\Psi(x,z)$, cf.~\cite[Eq.~(30)]{Marklof:2008dr}, which yields (see the last displayed equation of that paper) for $x\to\infty$ and $z\in(-1,1)$,
\begin{equation}\label{Psi2d}
\Psi(x,z) = \frac{3}{2\pi^2} (1-u)^2 x^{-1} + O(x^{-2})
\end{equation}
if $u:=x(1-|z|)\in[0,1)$, and
\begin{equation}
\Psi(x,z) = 0
\end{equation}
if $x(1-|z|) \notin[0,1)$.
The implied constant in \eqref{Psi2d} is independent of $x$ and $u$. The above asymptotics (and the fact that $\Psi(x,z)$ is uniformly bounded) imply
\begin{equation}
\overline\Psi(z):=\int_0^\infty \Psi(x,z)\,dx  = \log\frac{1}{1-|z|} + O(1),
\end{equation}
which holds uniformly for all $z\in(-1,1)$. We conclude that $\overline\Psi\in\L^1((-1,1),dz)$.
\end{proof}

\section{From discrete to continuous time}\label{sec:from}

The following proposition, together with Theorem \ref{thm:main3} (ii), immediately implies Theorem \ref{thm:main3} (i). 
Let us denote by
\begin{equation}
n_t := \big\lfloor \xibar^{-1} t \big\rfloor
\end{equation}
the (integer part of the) expected number of collisions within time $t$.

\begin{prop}
\label{prop:dt_ct_are_close}
For any $\vareps>0$ 
\begin{equation}
\label{dt_ct_are_close}
\frac{\|\vecX_t-\vecQ_{n_t}\|}{t^{5/12 + \vareps}} \toprob 0, 
\end{equation}
as $t\to\infty$.
\end{prop} 

By the same argument as in the proof of Theorem \ref{thm:main3} (see Section \ref{sec:general}), it is sufficient to prove Proposition \ref{prop:dt_ct_are_close} in the case when $(\xi_1,\veceta_1)$ has density $\Psi_0(x,\vecz)$. We will assume this from now on. Furthermore, note that the left hand side of \eqref{dt_ct_are_close} is independent of the choice of $\vecv_0$. We may therefore assume without loss of generality that $\vecv_0$ is a random variable uniformly distributed in $\US$ (this will only be used in the justification of rel.~\eqref{split_in_three2} below). The proof of Proposition \ref{prop:dt_ct_are_close}, which is given at the end of this section, exploits the following three lemmas.

\begin{lem}\label{lem:flight_time_clt_discrete_time}
\begin{equation}
\frac{\tau_n-n\xibar}{\sigma_d \sqrt{d\, n\log n}}
\Rightarrow
\scrN(0,1), 
\end{equation}
as $n\to\infty$.
\end{lem}

\begin{proof}
This is a simple variant of the proof of Theorem \ref{thm:main3} (ii).
\end{proof}

\begin{lem}
\label{lem:displacement_bound} 
For all $n\in\NN$ and $u\geq 2$,
\begin{equation}
\label{displacement_bound}
\prob{\|\vecQ_n\| > u} =  O\bigg(\frac{n \log u}{u^2}\bigg). 
\end{equation} 
\end{lem}

\begin{proof}
We begin by observing that
\begin{equation}\label{split_in_two_again}
\prob{\|\vecQ_n\| >u} \le \prob{\big\|\sum_{j=1}^n \xi^*_j \vecV_{j-1} \big\| >\frac{u}{2}}+\prob{\big\|\sum_{j=1}^n \mu_j \vecV_{j-1} \big\| >\frac{u}{2}},
\end{equation}
where $\xi^*_j=\xi_j-\mu_j$ and $\mu_j$ is the conditional expectation of the (untruncated) $\xi_j$ defined in \eqref{muj}. Recall also the definition of the corresponding conditional variance $\alpha_j$ in \eqref{alphaj}.
Now, 
\begin{equation}
\begin{split}
\prob{\big\|\sum_{j=1}^n \xi^*_j \vecV_{j-1} \big\| >\frac{u}{2}} 
& \leq \prob{\big\|\sum_{j=1}^n \xi^*_j \vecV_{j-1} \ind{\alpha_j\leq u} \big\| >\frac{u}{4}} + \prob{\big\|\sum_{j=1}^n \xi^*_j \vecV_{j-1} \ind{\alpha_j> u} \big\| >\frac{u}{4}} \\
& \leq \frac{16}{u^2} \sum_{j=1}^n  \expect{\alpha_j^2 \ind{\alpha_j\leq u}} + \sum_{j=1}^n \prob{\alpha_j> u} \\
& = O\bigg(\frac{n\log u}{u^2}\bigg) + O\bigg(\frac{n}{u^2}\bigg) ,
\end{split}
\end{equation}
where we have used Chebyshev's inequality and, in the last bound, Propositions \ref{betatail2} ($d=2$) and \ref{betatail3} ($d\geq 3$).

The second term in \eqref{split_in_two_again} is bounded similarly: We have
\begin{equation}\label{aftereight}
\begin{split}
\prob{\big\|\sum_{j=1}^n \mu_j \vecV_{j-1} \big\| >\frac{u}{2}} 
& \leq \prob{\big\|\sum_{j=1}^n \mu_j \vecV_{j-1} \ind{\mu_j\leq u} \big\| >\frac{u}{4}} + \prob{\big\|\sum_{j=1}^n \mu_j \vecV_{j-1} \ind{\mu_j> u} \big\| >\frac{u}{4}} \\
& \leq \frac{16}{u^2}   \expect{\big\|\sum_{j=1}^n \mu_j \vecV_{j-1} \ind{\mu_j\leq u} \big\|^2} + \sum_{j=1}^n \prob{\mu_j> u} .
\end{split}
\end{equation}
To control the first term on the right hand side of \eqref{aftereight}, it is sufficient to bound
\begin{equation}
\Exp\big[\big(\sum_{j=1}^n \zeta_j \big)^2\big] ,\qquad  \zeta_j:=(\vece\cdot \vecV_{j-1}) \mu_j  \ind{\mu_j\leq u} ,
\end{equation}
for arbitrary $\vece\in\US$. We follow the same steps as in the proof of Lemma \ref{mainlem1}.
Let us first show that 
\begin{equation}\label{expzeta}
\Exp(\zeta_j) = 
\begin{cases}
O( \omega^j \sqrt{\log\log u}\,) & (d=2)\\
O(\omega^j ) & (d\geq 3).
\end{cases}
\end{equation}
To this end choose in Proposition \ref{prop:exp_exp} $m=1$ and 
\begin{equation}
g(\vecw,\vecz)=\frac{K_{1}(\vecw,\vecz)}{K_0(\vecw,\vecz)} \ind{K_{1}(\vecw,\vecz) \leq u K_0(\vecw,\vecz)}.
\end{equation}
Propositions \ref{prop:tail2} and \ref{prop:tail3} yield
\begin{equation}
\Exp(g(\veceta_0,\veceta_1)^2)=\Exp(\mu_j^2 \ind{\mu_j\leq u})
=
\begin{cases}
O(\log\log u) & (d=2) \\
O(1) & (d\geq 3),
\end{cases}
\end{equation}
and hence \eqref{expzeta}.

The above implies
\begin{equation}
\var{\zeta_j^2} = \expect{\zeta_j^2} + 
\begin{cases}
O(\omega^{2j} \log\log u) & (d=2)\\
O(\omega^{2j}) & (d\geq 3),
\end{cases}
\end{equation}
where
\begin{equation}
\expect{\zeta_j^2} \leq \Exp(\mu_j^2 \ind{\mu_j\leq u}) = 
\begin{cases}
O(\log\log u)  &(d=2)\\
O(1) & (d\geq 3).
\end{cases}.
\end{equation}
Due to Proposition \ref{prop:exp_mix}, we also have
\begin{equation}
\cov{\zeta_i} {\zeta_j} \le \sqrt{\expect{\zeta_i^2}}\sqrt{\expect{\zeta_j^2}} \omega^{|i-j|}.
\end{equation}
We conclude
\begin{equation}
\label{sumXtight22}
\Exp\big[\big( \sum_{j=1}^n \zeta_j\big)^2\big]
= 
\begin{cases}
O(n\log\log u)  &(d=2)\\
O(n) & (d\geq 3).
\end{cases}
\end{equation}
The second term on the right hand side of \eqref{aftereight} is controlled by the tail estimates in Propositions \ref{prop:tail2} and \ref{prop:tail3}.  
The overall result is 
\begin{equation}
\prob{\big\|\sum_{j=1}^n \mu_j \vecV_{j-1} \big\| >\frac{u}{2}} =
\begin{cases}
\displaystyle
O\bigg(\frac{n\log\log u}{u^2}\bigg) + O\bigg(\frac{n}{u^2\log u}\bigg) & (d=2) \\[15pt]
\displaystyle
O\bigg(\frac{n}{u^2}\bigg) + O\bigg(\frac{n}{u^{1+\frac{d}{2}}}\bigg) & (d\geq 3) ,
\end{cases}
\end{equation}
which completes the proof of the lemma.
\end{proof}

\begin{lem}
\label{lem:fluctuation_bounds}
For any $\vareps>0$ and $\delta>0$,
\begin{equation}
\label{collision_no_fluct_bound}
\lim_{t\to\infty}
\prob{|\nu_t-n_t| > \delta t^{1/2 + \vareps}}
=0, 
\end{equation}
\begin{equation}
\label{max_flight_bound}
\lim_{n\to\infty}
\prob{\max_{1\le j \le n} \xi_j > \delta n^{1/2 + \vareps}}
=0,
\end{equation}
\begin{equation}
\label{max_displacement_bound}
\lim_{n\to\infty}
\prob{\max_{1\le m \le n} \| \vecQ_m \| > \delta n^{5/6 + \vareps}}
=0.
\end{equation}
\end{lem}

\begin{proof}[Proof of \eqref{collision_no_fluct_bound}]
Note that, for any $N\in\ZZ_{\geq 0}$, $t\geq 0$,
\begin{equation}
\nu_t \geq N \quad \Leftrightarrow \quad \tau_N\leq t,
\end{equation}
and therefore, with $N(t):=\lfloor n_t+\delta t^{1/2 + \vareps}\rfloor$,
\begin{equation}
\prob{\nu_t- n_t > \delta t^{1/2 + \vareps}}
\leq
\prob{\nu_t \geq N(t)}
=
\prob{\tau_{N(t)} \leq t} .
\end{equation}
On the other hand, it follows from Lemma \ref{lem:flight_time_clt_discrete_time} that 
\begin{equation}
\lim_{t\to\infty}
\prob{\tau_{N(t)} < t} 
= 0.
\end{equation}
Similarly, for $M(t):=\lfloor n_t- \delta t^{1/2 + \vareps}\rfloor$, we have
\begin{equation}
\prob{\nu_t- n_t < - \delta t^{1/2 + \vareps}}
=
\prob{\tau_{M(t)} > t},
\end{equation}
and Lemma \ref{lem:flight_time_clt_discrete_time} implies
\begin{equation}
\lim_{t\to\infty}
\prob{\tau_{M(t)} > t}
=
0.
\end{equation}
\end{proof}

\begin{proof}[Proof of \eqref{max_flight_bound}] 
We use the simplest union bound and Markov's inequality: 
\begin{equation}
\prob{\max_{1\le j \le n} \xi_j > \delta n^{1/2 + \vareps}}
\le 
n \prob{\xi_1 > \delta n^{1/2 + \vareps}}
\le 
n\frac{\expect{\xi_1^{2-\vareps}}}{\delta^{2-\vareps}n^{(1/2+\vareps)(2-\vareps)}}.
\end{equation}
This sequence converges to 0 for $\vareps<\frac32$.
\end{proof}

\begin{proof}[Proof of \eqref{max_displacement_bound}]
Note first that 
\begin{equation}
\max_{1\le m\le n}\|\vecQ_m\| 
\le 
\max_{1\le m\le n^{2/3}}\|\vecQ_{m\lfloor n^{1/3}\rfloor}\| + n^{1/3} \max_{1\le m\le n} \xi_m.
\end{equation}
Hence
\begin{multline}
\label{split_in_two}
\prob{\max_{1\le m \le n} \|\vecQ_m\| > \delta n^{5/6 + \vareps}} \\
\le 
\prob{\max_{1\le m\le n^{2/3}}\|\vecQ_{m\lfloor n^{1/3}\rfloor}\| > \frac{\delta}{2} n^{5/6 + \vareps}}
+
\prob{\max_{1\le m\le n} \xi_m > \frac{\delta}{2} n^{1/2 + \vareps}}.
\end{multline}
The second term on the right hand side of \eqref{split_in_two} converges to zero, due to \eqref{max_flight_bound}. From \eqref{displacement_bound} it follows that 
\begin{equation}
\prob{\max_{1\le m\le n^{2/3}}\|\vecQ_{m\lfloor n^{1/3}\rfloor}\| > \delta n^{5/6 + \vareps}}
\ll
\frac {\log n}{\delta^2 n^{5/3+2\vareps}} \sum_{m=1}^{\lfloor n^{2/3} \rfloor} m n^{1/3}
\ll
\frac {\log n}{\delta^2 n^{2\vareps}} \to 0. 
\end{equation}

This completes the proof of Lemma \ref{lem:fluctuation_bounds}.
\end{proof}

\begin{proof}[Proof of Proposition \ref{prop:dt_ct_are_close}]
Since 
\begin{equation}
\|\vecX_t-\vecQ_{n_t}\| \le \|\vecQ_{\nu_t}-\vecQ_{n_t}\| + \xi_{\nu_t+1}, 
\end{equation}
we have 
\begin{multline}
\label{split_in_three}
\prob{\|\vecX_t-\vecQ_{n_t}\| > \delta t^{5/12 + \vareps}} 
\le \prob{|\nu_t-n_t|> t^{1/2 + \vareps}} \\
+ \prob{\max_{|m|\le t^{1/2 + \vareps}} \xi_{n_t+m} > \frac{\delta}{2} t^{5/12 + \vareps}} 
+ \prob{\max_{|m|\le t^{1/2 + \vareps}} \|\vecQ_{n_t+m}-\vecQ_{n_t}\| > \frac{\delta}{2} t^{5/12 + \vareps}},
\end{multline}
and therefore, by stationarity of the Markov process \eqref{transprob} (recall that here we may assume without loss of generality that $\vecv_0$ is uniformly distributed in $\US$), 
\begin{multline}
\label{split_in_three2}
\prob{\|\vecX_t-\vecQ_{n_t}\| > \delta t^{5/12 + \vareps}} 
\leq \prob{|\nu_t-n_t|> t^{1/2 + \vareps}}\\
+
\prob{\max_{1\le m \le 2t^{1/2 + \vareps}} \xi_{m} > \frac{\delta}{2} t^{5/12 + \vareps}} 
+
\prob{\max_{1\le m \le 2t^{1/2 + \vareps}} \|\vecQ_{m}\| > \frac{\delta}{4} t^{5/12 + \vareps}} .
\end{multline}
The three terms on the right hand side of \eqref{split_in_three2} are controlled by Lemma \ref{lem:fluctuation_bounds}. 
This completes the proof of Proposition \ref{prop:dt_ct_are_close}.
\end{proof}

\section{Convergence of finite-dimensional distributions}\label{plussec:finite}

The convergence of finite-dimensional distribution follows from analogous arguments as in the one-dimensional case (cf.~Section \ref{sec:outline}). We include a sketch of the main steps. 

We will assume for the rest of this paper that $(\xi_1,\veceta_1)$ has density $\Psi_0(x, \vecz)$; by the arguments of Section \ref{sec:general} this is without loss of generality. 

\begin{prop}\label{plusprop:main3}
Let $d\geq 2$,  $\vecv_0\in\US$ and assume that the marginal distribution of $\veceta_1$ is absolutely continuous. Then, for every fixed $k$-tuple $(t_1,\ldots,t_k)\in(0,1]^k$ as $n\to\infty$,
\begin{equation}\label{pluseq:main4}
\big(\vecY_n(t_1),\ldots,\vecY_n(t_k) \big) \Rightarrow \big(\vecW(t_1),\ldots,\vecW(t_k)\big).
\end{equation}
\end{prop}

\begin{proof}
We may assume $t_0:=0<t_1<t_2<\ldots<t_k\leq 1$. The weak convergence \eqref{pluseq:main4} is equivalent to 
\begin{multline}\label{pluseq:main4b}
\big(\vecY_n(t_1)-\vecY_n(t_0),\ldots,\vecY_n(t_k)-\vecY_n(t_{k-1}) \big) \\ \Rightarrow \big(\vecW(t_1)-\vecW(t_0),\ldots,\vecW(t_k)-\vecW(t_{k-1})\big).
\end{multline}
Define the $kd$-dimensional vector
\begin{equation}
\vecU_{j,n} := \begin{pmatrix} 
\ind{j \leq  \lfloor t_1 n\rfloor} \vecV_j \\
\ind{\lfloor t_1 n\rfloor < j \leq  \lfloor t_2 n\rfloor} \vecV_j \\
\vdots \\
\ind{\lfloor t_{k-1} n\rfloor < j \leq  \lfloor t_k n\rfloor} \vecV_j
\end{pmatrix}
\end{equation}
and
\begin{equation}
\vecR_n := \sum_{j=1}^n \xi_j \vecU_j = 
\begin{pmatrix} \vecQ_n(t_1)-\vecQ_n(t_0) \\ \vdots \\ \vecQ_n(t_k)-\vecQ_n(t_{k-1}) \end{pmatrix}, \qquad n\in\NN .
\end{equation}
We thus need to show that
\begin{equation}\label{pluseq:main4c}
\frac{\vecR_n}{\sigma_d \sqrt{n\log n}} \Rightarrow \begin{pmatrix} \vecW_n(t_1)-\vecW_n(t_0) \\ \vdots \\ \vecW_n(t_k)-\vecW_n(t_{k-1}) \end{pmatrix}.
\end{equation}

We truncate $\vecR_n$ by defining the random variable
\begin{equation}
\vecR_n' := \sum_{j=1}^n \xi_j' \vecU_{j-1}  
\end{equation}
with $\xi_j'$ as in \eqref{xiprime}.
%The following lemma tells us that it is sufficient to prove Proposition \ref{plusprop:main3} for $\vecR_n'$ instead of $\vecR_n$.

\begin{lem}\label{plusmainlem0}
We have
\begin{equation}
\sup_{n\in\NN}\|\vecR_n-\vecR_n'\| <\infty
\end{equation}
almost surely.
\end{lem}

This statement is an immediate consequence of Lemma \ref{mainlem0}, where the bound is established for each component. %The argument uses the Borel-Cantelli Lemma, which shows that $\xi_j'\neq\xi_j$ occurs for finitely many $j$ almost surely.
To prove the central limit theorem for $\vecR_n'$, we consider $\txi_j = \xi_j' -m_j$
with the conditional expectation $m_j$ as in \eqref{conExp}, and let
\begin{equation}\label{plus2sums}
\widetilde\vecR_n := \sum_{j=1}^n \txi_j \vecU_{j-1} .
\end{equation}

%The following lemma shows that $\vecR_n'$ and $\widetilde\vecR_n$ are close relative to $\sqrt{n\log n}$.

\begin{lem}\label{plusmainlem1}
The sequence of random variables 
\begin{equation}
\frac{\vecR_n'-\widetilde\vecR_n}{\sqrt{n \log\log n}} 
\end{equation}
is tight if $d=2$, and  
\begin{equation}
\frac{\vecR_n'-\widetilde\vecR_n}{\sqrt{n}} 
\end{equation}
is tight if $d\geq 3$.
\end{lem}

This lemma follows directly from Lemma \ref{mainlem1}. %The latter exploits an upper bound on the diffusive order of fluctuations of ergodic averages for Markov chains, under appropriate spectral conditions.
Lemmas \ref{plusmainlem0} and \ref{plusmainlem1} imply that it is sufficient to prove Proposition \ref{plusprop:main3} for $\widetilde\vecR_n$ in place of $\vecR_n$. Let us turn to the covariance and recall the definition \eqref{cona} of $a_j$.

\begin{lem}\label{plusmainlem2}
For $n\to\infty$,
\begin{multline}\label{pluseq:mainlem2}
\frac{\Exp\big( \widetilde\vecR_n \otimes\widetilde\vecR_n\bigm| \uveceta \big)}{n\log n}  = \frac{\sum_{j=1}^n  a_j^2 \vecU_{j-1} \otimes\vecU_{j-1}}{n\log n}\\
\toprob  \sigma_d^2 \,  
\begin{pmatrix} 
t_1 I_d & & & \\
& (t_2-t_1) I_d & & \\
& & \ddots & \\
& & & (t_k-t_{k-1}) I_d
\end{pmatrix} .
\end{multline}
\end{lem}

This follows from Lemma \ref{mainlem2} by observing that the different $d$-dimensional components of $\widetilde\vecR_n$ are independent when conditioned on $\uveceta$ and $\vecv_0$. %The proof of Lemma \ref{mainlem2} establishes a weak law of large numbers for the sequence of conditional variances. This essentially relies on Chebyshev inequalities applied carefully to the truncated variables.
Note that the variance $A_n$ defined in \eqref{conA} satisfies $A_n^2 =\Exp( \|\widetilde\vecR_n\|^2)$.

%The next lemma \cite[Lemma 4.4]{super} verifies the Lindeberg conditions for random $\uveceta$. 
%%This is the main probabilistic ingredient in the proof of the central limit theorem and relies on delicate asymptotic estimates.
%
%\begin{lem}\label{plusmainlem3}
%For any fixed $\vareps>0$,
%\begin{equation}
%\label{plusLindeberg_condition}
%A_n^{-2}
%\sum_{j=1}^n
%\condexpect{\txi_j^2 \ind{\txi_j^2 > \vareps^2 A_n^2}}{\uveceta}
%\toprob 0
%\end{equation}
%as $n\to\infty$.
%\end{lem}

Given these lemmas, let us now conclude the proof of \eqref{pluseq:main4c}. The sequence of random vectors
\begin{equation}\label{pluseq:main4d}
\vecZ_n:=\frac{\widetilde\vecR_n}{\sigma_d \sqrt{n\log n}} 
\end{equation}
is tight in $\RR^{kd}$ because each component is tight in $\RR^d$; the latter was proved at the end of Section \ref{sec:outline}. By the Helly-Prokhorov theorem, there is an infinite subset $S_1\subset \NN$ so that $\vecZ_n$ converges in distribution along $n\in S_1$ to some limit $\vecZ$. Assume for a contradiction that $\vecZ$ is {\em not} the right hand side of \eqref{pluseq:main4c}. The Borel-Cantelli lemma implies that there is an infinite subset $S_2\subset S_1$, so that in the statements of Lemmas \ref{plusmainlem2} and \ref{mainlem3} we have almost-sure convergence along $n\in S_2$:
\begin{equation}\label{pluseq:mainlem222}
\Exp\big( \vecZ_n \otimes\vecZ_n\bigm| \uveceta \big)  \toas    \begin{pmatrix} 
t_1 I_d & & & \\
& (t_2-t_1) I_d & & \\
& & \ddots & \\
& & & (t_k-t_{k-1}) I_d
\end{pmatrix} ,
\end{equation}
\begin{equation}\label{pluseq:mainlem2bbb}
\frac{A_n^2}{n\log n}\toas  d\, \sigma_d^2 ,
\end{equation}
and
\begin{equation}
\label{plusLindeberg_condition22}
A_n^{-2}
\sum_{j=1}^n
\condexpect{\txi_j^2 \ind{\txi_j^2 > \vareps^2 A_n^2}}{\uveceta} 
\toas 0 .
\end{equation}
The hypotheses of the Lindeberg central limit theorem are thus satisfied, and we conclude the proof as for the one-dimensional distribution (see end of Section \ref{sec:outline}) with Lemmas \ref{mainlem0} and \ref{mainlem1} replaced by Lemmas \ref{plusmainlem0} and \ref{plusmainlem1}.
%
%and we infer that $\vecZ_n$ converges to the right hand side of \eqref{pluseq:main4c} for $n\to\infty$ along $S_2$. (We use the Lindeberg theorem for {\em triangular arrays} of independent random variables, since we have verified the Lindeberg conditions only along a subsequence.) This, however, contradicts our hypothesis, and hence the right hand side of \eqref{pluseq:main4c} is indeed the unique limit distribution of any weakly converging subsequence. This in turn implies that every sequence converges weakly, and therefore completes the proof of \eqref{pluseq:main4c}. In view of Lemmas \ref{plusmainlem0} and \ref{plusmainlem1}, this implies Proposition \ref{plusprop:main3}.
\end{proof}

\section{Tightness}\label{plussec:tightness}

We will now establish tightness for the sequence of processes $(\vecY_n)_{n=1}^\infty$. This is the last remaining input in the proof of Theorem \ref{plusthm:main2}. We assume for simplicity that the scattering map satisfies the hypotheses of Section \ref{sec:two_a} (as opposed to the milder condition in Theorem \ref{thm:main3}). 
Define
\begin{equation}
\xi_{j,n} := \xi_j \ind{\xi_j \leq   r_n},\qquad r_n=\sqrt{n (\log n)^\gamma},
\end{equation}
\begin{equation}
\txi_{j,n} = \xi_{j,n} -m_{j,n},  
\end{equation}
with the conditional expectation
\begin{equation}\label{pluscexp}
m_{j,n} := \Exp\big( \xi_{j,n} \bigm| \uveceta\big) = \frac{K_{1,r_n}(\veceta_{j-1},\veceta_j)}{K_0(\veceta_{j-1},\veceta_j)} .
\end{equation}
Furthermore, let
\begin{equation}
\vecQ_n^* := \sum_{j=1}^n \xi_{j,n} \vecV_{j-1} ,
\end{equation}
\begin{equation}
a_{j,n}^2:= \Var\big(\xi_{j,n} \bigm| \uveceta\big)= \frac{K_{2,r_n}(\veceta_{j-1},\veceta_j)}{K_0(\veceta_{j-1},\veceta_j)} - m_{j,n}^2 ,
\end{equation}
and
\begin{equation}
\scrA_n^2 :=\sum_{j=1}^n a_{j,n}^2 = \Exp\big( \|\vecQ_n^*\|^2 \bigm| \uveceta \big).
\end{equation}

We split the process $\vecY_n$ defined in \eqref{eq:main4b} into four parts,
\begin{equation}
\vecY_n= \widehat\vecY_n + \widetilde\vecY_n + \overline\vecY_n + \check\vecY_n,
\end{equation}
where
\begin{equation}
\widehat\vecY_n(t)
 = \frac{1}{\sigma_d \sqrt{n\log n}}\sum_{j=1}^{\lfloor nt \rfloor} \xi_j \ind{\xi_j >   r_n}\vecV_{j-1},
\end{equation}
\begin{equation}
\widetilde\vecY_n(t)
 = \frac{1}{\sigma_d \sqrt{n\log n}}\sum_{j=1}^{\lfloor nt \rfloor} \txi_{j,n} \vecV_{j-1},
\end{equation}
\begin{equation}
\overline\vecY_n(t)
 = \frac{1}{\sigma_d \sqrt{n\log n}}  \sum_{j=1}^{\lfloor nt \rfloor} m_{j,n} \vecV_{j-1} ,
\end{equation}
and
\begin{equation}
\check\vecY_n(t)
 = \frac{1}{\sigma_d \sqrt{n\log n}}   \{ nt\}\, \xi_{\lfloor nt\rfloor +1}\vecV_{\lfloor nt\rfloor}.
\end{equation}

We begin by showing that $\widehat\vecY_n$, $\overline\vecY_n$ and  $\check\vecY_n$ are uniformly small with large probability. Consider first $\widehat\vecY_n$ and $\check\vecY_n$.

\begin{prop}\label{pluslemY0}
We have 
\begin{equation}
\sup_{n\in\NN} \sup_{t\in[0,1]} \sqrt{n\log n} \; \| \widehat\vecY_n(t)\| <\infty 
\end{equation}
almost surely.
\end{prop}

The proof of Proposition \ref{pluslemY0} is identical to that of Lemma \ref{mainlem0}. The key ingredient is the Borel-Cantelli Lemma, cf.~the comment after Lemma \ref{plusmainlem0}.

\begin{prop}\label{pluslemY1}
There is $C<\infty$ such that, for any $\beta>0$ and any $n\geq 2$, 
\begin{equation}
\Prob \bigg(\sup_{t\in [0,1]} \|\check\vecY_n(t)\|  \geq \beta  \bigg) \leq \frac{C}{\beta^2 \log n}.
\end{equation}
\end{prop}

\begin{proof}
We have
\begin{equation}
\begin{split}
\Prob \bigg(\sup_{t\in [0,1]} \|\check\vecY_n(t)\|  \geq \beta  \bigg)
& \leq \Prob \bigg(\max_{1\leq j \leq n+1} |\xi_j|  \geq \beta\sigma_d \sqrt{n\log n}  \bigg) \\
& \leq  \sum_{j=1}^{n+1} \Prob \bigg(|\xi_j|  \geq \beta\sigma_d \sqrt{n\log n}  \bigg) 
\end{split}
\end{equation}
where, by the asymptotic tail for the free path length distribution \cite[Theorem 1.14]{Marklof:2011di}, we have for all $j\geq 1$ 
\begin{equation}
\Prob \bigg(|\xi_j|  \geq \beta\sigma_d \sqrt{n\log n}  \bigg) =O\bigg( \frac{1}{\beta^2 n \log n} \bigg).
\end{equation}
\end{proof}

The estimation of $\overline\vecY_n(t)$ relies on the following maximal inequality for martingales with stationary increments, cf.~Gordin and Lifsic \cite{Gordin78}. Let $(\scrV,\mu)$ be a probability space. We denote by $\L_0^2(\scrV,\mu)$ the orthogonal complement of the constant functions in $\L^2(\scrV,\mu)$.

\begin{prop}\label{plusGordin}
Let $\uvecalf=(\vecalf_n)_{n=0}^\infty$ be a Markov chain on the state space $\scrV$, and $\mu$ a probability measure which is stationary and ergodic for $\uvecalf$. Let $\scrP$ be the transition operator on $\L^2(\scrV,\mu)$ defined by $\scrP f(\vecz) = \Exp\big(f(\vecalf_n) \bigm| \vecalf_{n-1}=\vecz \big)$.
Then, for any $f\in\Ran(\scrI-\scrP)$, $n\in\NN$, $\kappa>0$, 
\begin{equation}\label{pluseq:Gordin}
\Prob\bigg(\max_{1\leq m \leq n} \bigg| \sum_{j=1}^m f(\vecalf_j) \bigg| \geq \kappa \sqrt{n} \bigg) \leq \frac{9}{\kappa^2} \bigg( \| g \|^2 +\frac1n\, \|\scrP g \|^2 \bigg) ,
\end{equation}
where $g$ is the unique function in $\L_0^2(\scrV,\mu)$ such that $f=(\scrI-\scrP)g$.
\end{prop}

\begin{proof}
We have
\begin{equation}
\sum_{j=1}^m f(\vecalf_j) = \sum_{j=1}^m \big(g(\vecalf_j)-\scrP g(\vecalf_j)\big) = \scrM_m + \scrP g(\vecalf_0) - \scrP g(\vecalf_m) .  
\end{equation}
where
\begin{equation}
\scrM_m:=\sum_{j=1}^m \big(g(\vecalf_j)-\scrP g(\vecalf_{j-1})\big)
\end{equation}
is a martingale with stationary and ergodic increments. The left hand side of \eqref{pluseq:Gordin} is estimated by the sum of the following three terms. The first is bounded by Doob's inequality for non-negative sub-martingales,
\begin{equation}\label{pluseq:Gordin1}
\Prob\bigg(\max_{1\leq m \leq n} \big| \scrM_m \big| \geq \frac{\kappa \sqrt{n}}{3} \bigg) \leq \frac{9}{\kappa^2 n} \Exp\big( \big| \scrM_n \big|^2 \big) =\frac{9}{\kappa^2} \big(\|g\|^2-\|\scrP g\|^2 \big) .
\end{equation}
The second follows from Chebyshev's inequality,
\begin{equation}\label{pluseq:Gordin2}
\Prob\bigg(\big| \scrP g(\vecalf_0) \big| \geq \frac{\kappa \sqrt{n}}{3} \bigg) \leq \frac{9}{\kappa^2 n}\|\scrP g\|^2  ,
\end{equation}
and the third from the union bound and Chebyshev's inequality,
\begin{equation}\label{pluseq:Gordin3}
\Prob\bigg(\max_{1\leq m \leq n} \big| \scrP g(\vecalf_m) \big| \geq \frac{\kappa \sqrt{n}}{3} \bigg) \leq
\sum_{m=1}^n\Prob\bigg(\big| \scrP g(\vecalf_m) \big| \geq \frac{\kappa \sqrt{n}}{3} \bigg) \leq  \frac{9}{\kappa^2}\|\scrP g\|^2  .
\end{equation}
\end{proof}

\begin{prop}\label{pluslemY2}
There is $C<\infty$ such that, for any $\beta>0$ and any $n\geq 2$,
\begin{equation}
\Prob \bigg(\sup_{t\in [0,1]} \|\overline\vecY_n(t)\|  \geq \beta  \bigg) 
\leq 
\begin{cases}
\frac{C\log\log n}{\beta^2 \log n}  & (d=2) \\[5pt]
\frac{C}{\beta^2 \log n}& (d\geq 3).
\end{cases}
\end{equation}
\end{prop}

\begin{proof}
The plan is to apply Proposition \ref{plusGordin} to a Markov chain defined on the state space of three consecutive velocities,
\begin{equation}
\scrV:= \{ (\vecv_{n-1},\vecv_n,\vecv_{n+1}) \in (\US)^3 : \varphi(\vecv_{n-1},\vecv_n)>B_\theta,\;  \varphi(\vecv_n,\vecv_{n+1})>B_\theta \} ,
\end{equation}
where $\varphi(\vecu_1,\vecu_2)\in[0,\pi]$ denotes the angle between the two vectors $\vecu_1,\vecu_2$, and $B_\theta$ as in \eqref{Btheta}. For $(\vecv_{n-1},\vecv_n,\vecv_{n+1})\in\scrV$ and $\xi>0$, let
\begin{equation}
p_0(\vecv_{n-1},\vecv_n,\xi,\vecv_{n+1}) = \Psi_0\big(\vecw(\vecv_{n-1},\vecv_n),\xi, \vecw'(\vecv_n,\vecv_{n+1})\big)\,\sigma(\vecv_n,\vecv_{n+1})
\end{equation}
where  the functions $\vecw$, $\vecw'$ are defined via
\begin{equation}
 \begin{pmatrix} 0 \\ \vecw(\vecv_{n-1},\vecv_n) \end{pmatrix}= 
R(\vecv_n)^{-1} \vecs(\vecv_{n-1},\vecv_n) ,
\end{equation}
\begin{equation}
\begin{pmatrix} 0 \\ \vecw'(\vecv_n,\vecv_{n+1}) \end{pmatrix} 
= R(\vecv_n)^{-1} \vecb(\vecv_n,\vecv_{n+1}) ;
\end{equation}
here $\vecs(\vecv_{n-1},\vecv_n)$, $\vecb(\vecv_n,\vecv_{n+1})\in\RR^d$ are the exit and impact parameters, respectively, expressed as functions of incoming and outgoing velocities. Furthermore $\sigma(\vecv_0,\vecv)$ is the differential cross section, i.e.\ the differential of the map $\vecv \mapsto \vecb(\vecv_0,\vecv)$. The function $p_0(\vecv_{n-1},\vecv_n,\xi,\vecv_{n+1})$ is precisely the transition kernel that governs the Boltzmann-Grad limit of the periodic Lorentz gas in the velocity representation used in \cite{Marklof:2011ho}. We integrate out the flight time and obtain
\begin{equation}
\int_0^\infty p_0(\vecv_{n-1},\vecv_n,\xi,\vecv_{n+1})\, d\xi = L_0(\vecv_{n-1},\vecv_n,\vecv_{n+1})\,\sigma(\vecv_n,\vecv_{n+1}) 
\end{equation}
with
\begin{equation}
L_0(\vecv_{n-1},\vecv_n,\vecv_{n+1}):=K_0\big(\vecw(\vecv_{n-1},\vecv_n), \vecw'(\vecv_n,\vecv_{n+1}) \big),
\end{equation}
with the kernel $K_0(\vecw,\vecz)$ defined in \eqref{K0def}.
Thus
\begin{equation}
n\mapsto\vecalf_n = (\vecV_{n-1},\vecV_n,\vecV_{n+1})
\end{equation}
(where $\vecV_j$ is the random variable defined in \eqref{vac}) defines a Markov chain on the state space $\scrV$ with stationary measure
\begin{equation}
d\mu((\vecv_{n-1},\vecv_n,\vecv_{n+1})) := {\overline s}^{-1} \sigmabar^{-2}\; \sigma(\vecv_{n-1},\vecv_n)\,\sigma(\vecv_n,\vecv_{n+1})\; d\vecv_{n-1}\,d\vecv_n\,d\vecv_{n+1} 
\end{equation}
where ${\overline s}$ is the volume of $\US$ and $\sigmabar$ the volume of $\UB$.
Explicitly, for $\scrA\subset\scrV$,
\begin{equation}
\prob{ \vecalf_n \in\scrA  \bigm| \vecalf_{n-1}=\vecz } 
=\int_\scrA  \scrK( \vecz,\vecz' )\, d\mu(\vecz')
\end{equation}
with kernel
\begin{multline}
\scrK\big( (\vecv_{n-2},\vecv_{n-1},\vecv_n),(\vecv_{n-1}',\vecv_n',\vecv_{n+1}')
\big) \\
=
{\overline s}\, \sigmabar^2 \; \frac{\delta(\vecv_{n-1},\vecv_{n-1}')\, \delta(\vecv_n,\vecv_n') \, L_0(\vecv_{n-1}',\vecv_n',\vecv_{n+1}')}{\sigma(\vecv_{n-1},\vecv_n)} .
\end{multline}
Note that the kernel $\scrK(\vecz,\vecz')$ is defined with respect to the measure $\mu$ and {\em not} with respect to the Lebesgue measure.
  
The transition operator $\scrP$ for this Markov chain is defined by
\begin{equation}
\begin{split}
\scrP f(\vecz) & = \Exp\big(f(\vecalf_n) \bigm| \vecalf_{n-1}=\vecz \big) \\
& = \int \scrK(\vecz,\vecz') f(\vecz') \,d\mu(\vecz') .
\end{split}
\end{equation}
The operator $\scrP$ has eigenvalue $1$ (corresponding to constant eigenfunctions). To establish that there is a spectral gap, we calculate the kernel $\scrK^{(m)}(\vecz,\vecz')$ of the $m$th power,
\begin{equation}
\begin{split}
\scrP^m f(\vecz) & = \Exp\big(f(\vecalf_n) \bigm| \vecalf_{n-1}=\vecz \big) \\
& = \int \scrK^{(m)}(\vecz,\vecz') f(\vecz') \,d\mu(\vecz') .
\end{split}
\end{equation}
The second power reads
\begin{multline}
\scrK^{(2)}\big( (\vecv_{n-2},\vecv_{n-1},\vecv_n),(\vecv_{n}',\vecv_{n+1}',\vecv_{n+2}')
\big) \\
=  
{\overline s}\, \sigmabar^2 L_0(\vecv_{n-1},\vecv_n,\vecv_{n+1}') \, 
L_0(\vecv_{n}',\vecv_{n+1}',\vecv_{n+2}') \,\delta(\vecv_n,\vecv_n') .
\end{multline}
As to the third power,
\begin{multline}
\scrK^{(3)}\big( (\vecv_{n-2},\vecv_{n-1},\vecv_n),(\vecv_{n+1}',\vecv_{n+2}',\vecv_{n+3}')
\big) \\
=  
{\overline s}\, \sigmabar^2 L_0(\vecv_{n-1},\vecv_n,\vecv_{n+1}') \, 
L_0(\vecv_{n},\vecv_{n+1}',\vecv_{n+2}')\,
L_0(\vecv_{n+1}',\vecv_{n+2}',\vecv_{n+3}')\,
\sigma(\vecv_n,\vecv_{n+1}') .
\end{multline}
In view of the lower bound on $K_0(\vecw,\vecz)$ (Lemmas \ref{six-one}, \ref{K0upperlower}), we have
\begin{equation}
\scrK^{(3)}\big( (\vecv_{n-2},\vecv_{n-1},\vecv_n),(\vecv_{n+1}',\vecv_{n+2}',\vecv_{n+3}')
\big) \geq \frac{{\overline s}\, \sigmabar^2 }{(2^d \sigmabar\zeta(d))^3} 
\;\sigma(\vecv_n,\vecv_{n+1}').
\end{equation}
The fourth power reads 
\begin{multline}
\scrK^{(4)}\big( (\vecv_{n-2},\vecv_{n-1},\vecv_n),(\vecv_{n+2}',\vecv_{n+3}',\vecv_{n+4}')
\big) \\
=  
{\overline s}\, \sigmabar^2 \int_\US L_0(\vecv_{n-1},\vecv_n,\tilde\vecv_{n+1}) \, 
L_0(\vecv_{n},\tilde\vecv_{n+1},\vecv_{n+2}')\,
L_0(\tilde\vecv_{n+1},\vecv_{n+2}',\vecv_{n+3}')\\
\times L_0(\vecv_{n+2}',\vecv_{n+3}',\vecv_{n+4}') \,
\sigma(\vecv_n,\tilde\vecv_{n+1})\,\sigma(\tilde\vecv_{n+1},\vecv_{n+2}') \, d\tilde \vecv_{n+1} .
\end{multline}
The lower bound on $K_0(\vecw,\vecz)$ now yields
\begin{multline}
\scrK^{(4)}\big( (\vecv_{n-2},\vecv_{n-1},\vecv_n),(\vecv_{n+2}',\vecv_{n+3}',\vecv_{n+4}')
\big) \\ \geq \frac{{\overline s}\, \sigmabar^2 }{(2^d \sigmabar\zeta(d))^4} 
\int_\US 
\sigma(\vecv_n,\tilde\vecv_{n+1})\,\sigma(\tilde\vecv_{n+1},\vecv_{n+2}') \, d\tilde \vecv_{n+1} .
\end{multline}
Similarly, for the $m$th power ($m\geq 5$),
\begin{multline}
\scrK^{(m)}\big( (\vecv_{n-2},\vecv_{n-1},\vecv_n),(\vecv_{n+m-2}',\vecv_{n+m-1}',\vecv_{n+m}')
\big) \\ \geq \frac{{\overline s}\, \sigmabar^2 }{(2^d \sigmabar\zeta(d))^m} 
\int_{(\US)^{m-3}} 
\sigma(\vecv_n,\tilde\vecv_{n+1})\, \sigma(\tilde\vecv_{n+1},\tilde\vecv_{n+2})\cdots \\
\cdots \sigma(\tilde\vecv_{n+m-4},\tilde\vecv_{n+m-3}) \,\sigma(\tilde\vecv_{n+m-3},\vecv_{n+m-2}') \, d\tilde \vecv_{n+1}\cdots d\tilde \vecv_{n+m-3}.
\end{multline}
For the class of scattering maps considered in Section \ref{sec:two_a}, there exists a finite $m$ such that the integral on the right-hand side has a uniform positive lower bound for all $\vecv_n,\vecv_{n+m-2}'\in\US$. Hence
\begin{equation}
\inf_{\vecz,\vecz'\in\scrV} \scrK^{(m)}(\vecz,\vecz') >0.
\end{equation}
This, by the classic Doeblin argument (recall Section \ref{sec:spectral}), implies that $\scrP^m$, and therefore $\scrP$, has a spectral gap.

Let
\begin{equation}
L_{1,r}(\vecv_{n-1},\vecv_n,\vecv_{n+1}):=K_{1,r} \big(\vecw(\vecv_{n-1},\vecv_n), \vecw'(\vecv_n,\vecv_{n+1}) \big),
\end{equation}
with $K_{1,r}(\vecw,\vecz)$ as in \eqref{K1rdef},
and define the function $\ell_r:\scrV\to \RR_{\geq 0}$ by
\begin{equation}
\ell_r((\vecv_{n-1},\vecv_n,\vecv_{n+1})):=\frac{L_{1,r}(\vecv_{n-1},\vecv_n,\vecv_{n+1})}{L_0(\vecv_{n-1},\vecv_n,\vecv_{n+1})}.
\end{equation}
We then recover the random variable \eqref{pluscexp} via
\begin{equation}
m_{j,n} = \ell_{r_n}(\vecalf_n) = \ell_{r_n}((\vecV_{n-1},\vecV_n,\vecV_{n+1})).
\end{equation}
To conclude the proof, apply Proposition \ref{plusGordin} with the ($n$-dependent) choice of $f\in\L^2(\scrV,d\mu)$,
\begin{equation}
f((\vecv_{n-1},\vecv_n,\vecv_{n+1}))=\vece\cdot \vecv_n\, \ell_{r_n}((\vecv_{n-1},\vecv_n,\vecv_{n+1}))
\end{equation}
with an arbitrary fixed $\vece\in\US$. Since $\ell_{r_n}((R\vecv_{n-1},R\vecv_n,R\vecv_{n+1}))=\ell_{r_n}((\vecv_{n-1},\vecv_n,\vecv_{n+1}))$ and $d\mu((R\vecv_{n-1},R\vecv_n,R\vecv_{n+1}))=d\mu((\vecv_{n-1},\vecv_n,\vecv_{n+1}))$ for all $R\in\SO(d)$, we have
\begin{equation}
\int_\scrV f(\vecz)\,d\mu(\vecz)=0 ,
\end{equation}
i.e., $f\in\L_0^2(\scrV,\mu)$. Thanks to the spectral gap of $\scrP$, there is a constant $M<\infty$, such that for all $f\in\L_0^2(\scrV,\mu)$,
\begin{equation}
\| (\scrI-\scrP)^{-1} f \|^2 \leq M \| f \|^2
\end{equation}
Finally, the estimate for $\Exp(m_n^2)$ in the proof of Lemma \ref{mainlem1} yields $\| f \|^2 = O(\log\log n)$ if $d=2$ and $\| f \|^2 = O(1)$ if $d\geq 3$.
Proposition \ref{pluslemY2} thus follows from Proposition \ref{plusGordin} with $\kappa=\beta\sqrt{\log n}$.
\end{proof}

Propositions \ref{pluslemY0}, \ref{pluslemY1} and \ref{pluslemY2} establish that the tightness for $(\vecY_n)_{n=1}^\infty$ is implied by the tightness for $(\widetilde\vecY_n)_{n=1}^\infty$. To prove the latter, we use the following classical characterization of tightness for a random curve $\vecX$.

\begin{thm}\label{plusthm:Bilingsley} \cite[Theorem 8.3]{Billingsley68}
The sequence $(\lambda_n)_{n=1}^\infty$ of probability measures in $\C_0([0,1])$ is tight if for every $\epsilon>0$, $\beta>0$ there exist $\delta<1$, $n_0<\infty$ such that for all $n\geq n_0$ we have
\begin{equation}
\frac{1}{\delta} \sup_{t\in[0,1]} \lambda_n \bigg( \sup_{s\in[t,t+\delta]\cap[0,1]} \big\| \vecX(s)-\vecX(t) \big\| \geq \beta \bigg) <\epsilon.
\end{equation}
\end{thm}

We will also exploit the following maximal inequality for sums of independent random variables.

\begin{lem} \label{pluslem:max}
\cite[Lemma, p.~69]{Billingsley68}
Let $\vecxi_1,\ldots,\vecxi_m$ be independent random variables in $\RR^d$ with mean zero and finite variances $\sigma_i^2=\Exp(\|\vecxi_i\|^2)$. Put $\vecS_m= \vecxi_1+\ldots+\vecxi_m$ and $s_m^2=\sigma_1^2+\ldots+\sigma_m^2$. Then, for any $\lambda\in\RR$,
\begin{equation}\label{pluseq:max}
\Prob\bigg( \max_{i\leq m} \|\vecS_i \| \geq \lambda s_m\bigg) \leq 2 \Prob\big( \|\vecS_m \|\geq (\lambda-\sqrt2) s_m \big).
\end{equation}
\end{lem}

The following proposition verifies the hypothesis of Theorem \ref{plusthm:Bilingsley}, and thus proves that the sequence of probability measures corresponding to $(\widetilde\vecY_n)_{n=1}^\infty$ is tight.

\begin{prop}
For every $\epsilon>0$, $\beta>0$ there exist $\delta<1$, $n_0<\infty$ such that for all $n\geq n_0$ we have
\begin{equation}
\frac{1}{\delta} \sup_{t\in[0,1]} \Prob \bigg( \sup_{s\in[t,t+\delta]\cap[0,1]} \big\| \widetilde\vecY_n(s)-\widetilde\vecY_n(t) \big\| \geq \beta \bigg) <\epsilon.
\end{equation}
\end{prop}

\begin{proof}
We need to show that, for every $\epsilon>0$, $\beta>0$ there exist $\delta<1$, $n_0<\infty$ such that for all $n\geq n_0$ we have
\begin{equation}
\frac{1}{\delta}\Prob \bigg( \max_{nt< m \leq n(t+\delta)} \bigg\| \sum_{j=\lfloor nt \rfloor+1}^m \txi_{j,n} \vecV_{j-1} \bigg\| \geq \beta \sqrt{n\log n}\bigg) <\epsilon . 
\end{equation}
To prove this fact, note first that
\begin{multline}
\Prob \bigg( \max_{nt< m \leq n(t+\delta)} \bigg\| \sum_{j=\lfloor nt \rfloor+1}^m \txi_{j,n} \vecV_{j-1} \bigg\| \geq \beta \sqrt{n\log n}\bigg) \\
= \Exp\bigg( \Prob \bigg( \max_{nt< m \leq n(t+\delta)} \bigg\| \sum_{j=\lfloor nt \rfloor+1}^m \txi_{j,n} \vecV_{j-1} \bigg\| \geq \beta \sqrt{n\log n}\biggm| \uveceta \bigg) \bigg) .
\end{multline}
The maximal inequality \eqref{pluseq:max} yields
\begin{multline}
\Prob \bigg( \max_{nt< m \leq n(t+\delta)} \bigg\| \sum_{j=\lfloor nt \rfloor+1}^m \txi_{j,n} \vecV_{j-1} \bigg\| \geq \beta \sqrt{n\log n}\biggm| \uveceta \bigg)\\
\leq 2 \Prob \bigg(  \bigg\| \sum_{j=\lfloor nt \rfloor+1}^{\lfloor n(t+\delta)\rfloor} \txi_{j,n} \vecV_{j-1} \bigg\| \geq \beta  \sqrt{n\log n} - \sqrt{2(\scrA_{\lfloor n(t+\delta)\rfloor}^2-\scrA_{\lfloor nt\rfloor}^2)} \biggm| \uveceta \bigg),
\end{multline}
which in turn implies, after taking expectation values, 
\begin{multline}
\Prob \bigg( \max_{nt< m \leq n(t+\delta)} \bigg\| \sum_{j=\lfloor nt \rfloor+1}^m \txi_{j,n} \vecV_{j-1} \bigg\| \geq \beta \sqrt{n\log n}\bigg)\\
\leq 2 \Prob \bigg(  \bigg\| \sum_{j=\lfloor nt \rfloor+1}^{\lfloor n(t+\delta)\rfloor} \txi_{j,n} \vecV_{j-1} \bigg\| \geq \beta \sqrt{n\log n} - \sqrt{2(\scrA_{\lfloor n(t+\delta)\rfloor}^2-\scrA_{\lfloor nt\rfloor}^2)} \bigg).
\end{multline}
The latter is bounded above by
\begin{multline}\label{plustoto}
\leq 2 \Prob \bigg(  \bigg\| \sum_{j=\lfloor nt \rfloor+1}^{\lfloor n(t+\delta)\rfloor} \txi_{j,n} \vecV_{j-1} \bigg\| \geq \big(\beta-K\sqrt\delta \big) \sqrt{n\log n}  \bigg) \\
+ 2 \Prob\big( \scrA_{\lfloor n(t+\delta)\rfloor}^2-\scrA_{\lfloor nt\rfloor}^2 > K^2\delta n\log n \big) ,
\end{multline}
for any constant $K\geq 0$. By adapting the proof of Lemma \ref{mainlem2} (replacing $j(\log j)^\gamma$ by $n(\log n)^\gamma$), one can prove that the analogue of \eqref{eq:mainlem2b} reads
\begin{equation}\label{pluseq:mainlem2bb2}
\frac{\scrA_n^2}{n\log n}\toprob  d\,\sigma_d^2 .
\end{equation}
This implies that for any constant $K>\sigma_d\sqrt d$ we have
\begin{equation}
\lim_{n\to\infty} \Prob\big( \scrA_{\lfloor n(t+\delta)\rfloor}^2-\scrA_{\lfloor nt\rfloor}^2 > K^2\delta\, n\log n \big) = 0
\end{equation}
uniformly in $t$, $\delta$, which takes care of the second term in \eqref{plustoto}. The first term in \eqref{plustoto} is estimated by Theorem \ref{thm:main3} (ii): Given $\beta,\epsilon$, for any sufficiently small $\delta$ there is $n_0$ such that for all $t\in[0,1]$ and $n\geq n_0$
\begin{multline}\label{plustoto2}
\Prob \bigg(  \bigg\| \sum_{j=\lfloor nt \rfloor+1}^{\lfloor n(t+\delta)\rfloor} \txi_{j,n} \vecV_{j-1} \bigg\| \geq \big(\beta-K\sqrt\delta \big) \sqrt{n\log n}  \bigg) \\
\leq \frac{1}{(2\pi)^{d/2}} \int_{\|\vecx\|>\frac{\beta-K\sqrt\delta}{\sigma_d \sqrt\delta}} \e^{-\frac12 \|\vecx\|^2} d\vecx +\epsilon\,\delta = O(\epsilon\,\delta) .
\end{multline}
Note that here we have applied Theorem \ref{thm:main3} (ii) to the truncated $\txi_{j,n}$ rather than $\xi_j$, which is justified by the analogue of Lemmas \ref{mainlem0}, \ref{mainlem1} with $j(\log j)^\gamma$ replaced by $n(\log n)^\gamma$. This completes the proof.
\end{proof}

\section{Theorem \ref{plusthm:main2} implies Theorem \ref{plusthm:main1}}\label{plussec:221}

We now turn to the continuous time process. The Boltzmann-Grad limit $r\to 0$ is covered by \cite[Theorem 1.2]{Marklof:2011ho}, which tells us that for arbitrary fixed $T$,
\begin{equation}\label{pluseq:main2345678}
\vecX_{T,r}\Rightarrow \vecX_T 
\end{equation}
where
\begin{equation}
\vecX_T(t)=\frac{\vecX(t T)}{\Sigma_d\sqrt{T\log T}} =  \frac{\vecQ_{\nu_{tT}} + (t-\tau_{\nu_{tT}})\vecV_{\nu_{tT}}} {\Sigma_d\sqrt{T\log T}} ;
\end{equation}
recall \eqref{set1}--\eqref{set2}.

The convergence of finite-dimensional distributions of $\vecX_T$ to $\vecW$ follows (within the framework of Section \ref{plussec:finite}) from the same estimates as in Section \ref{sec:from}. What remains is to show tightness in $\C([0,1])$ for the family of processes $(\vecX_T)_{T\geq 1}$. By a simple scaling argument, it is sufficient to prove tightness for the sequence $(\vecX_n)_{n\in\NN}$. 

Define the continuous, strictly increasing (random) functions $T,\Theta : \RR_{\geq 0}\to\RR_{\geq 0}$ by
\begin{equation}
T(\theta):=\tau_{\lfloor\theta\rfloor} + \{\theta\} \xi_{\lfloor\theta\rfloor+1},\qquad
\Theta(t):=\nu_t+\frac{t-\tau_{\nu_t}}{\xi_{\nu_t+1}}.
\end{equation}
Note that the functions $T$ and $\Theta$ are the inverse of one another. That is, $T(\Theta(t))=t$ and $\Theta(T(\theta))=\theta$. The key point is that we can now relate the curves $t\to\vecX_n(t)$ and $\theta\to\vecQ_n(\theta)$ by this time change: We have $\vecX(t)=\vecQ_1(\Theta(t))$
and therefore
\begin{equation}
\vecX(n t)=\vecQ_1(\Theta(nt)) =\vecQ_n(n^{-1}\Theta(nt)) =\vecQ_n(\Theta_n(t)),
\end{equation}
where $\Theta_n(t):=n^{-1} \Theta(n t)$.
This yields for the normalized processes
\begin{equation}
\vecX_n(t)=\xibar^{1/2} \, \vecY_n(\Theta_n(t)) ,
\end{equation}
where $\xibar=1/\sigmabar$ is the mean free path length.

Given $b>0$, consider the random process $Z_n : [0,b] \to \RR$ defined by
\begin{equation}
Z_n(\theta):=\frac{T(n\theta)-n\theta\xibar}{\sigma_d\sqrt{d\,  n\log n}} .
\end{equation}
The following lemma says that $Z_n$ converges to one-dimensional Brownian motion $W$.
%; cf.~also \cite[Lemma 11.3]{super}.

\begin{lem}\label{pluslem:Zn}
For $n\to\infty$,
\begin{equation}
Z_n \Rightarrow W.
\end{equation}
\end{lem}

The proof of this lemma is a simpler variant of the already established weak convergence $\vecY_n\Rightarrow\vecW$.

The lemma implies in particular that the process $\theta\mapsto T_n(\theta):=n^{-1} T(n\theta)$ converges weakly to the deterministic function $\theta\mapsto \xibar\theta$. Since $T_n(\Theta_n(t))=t$ and $\Theta_n(T_n(\theta))=\theta$, this implies that $t\mapsto \Theta_n(t)$ converges weakly to $t\mapsto \sigmabar t$.

The modulus of continuity of a curve $\vecX\in\C_0([0,b])$ is defined as
\begin{equation}
\omega_\vecX^{[0,b]}(\delta) := \sup_{\substack{0\leq s,t\leq b\\ |t-s|\leq \delta}}\big\| \vecX(s)-\vecX(t) \| .
\end{equation}
The tightness for $(\vecX_n)_{n\in\NN}$ is implied by the following lemma.

\begin{lem}
For every $\beta>0$, $\epsilon>0$ there exist $\delta<1$ and $n_0<\infty$ such that for all $n\geq n_0$
\begin{equation}
\Prob\big( \omega_{\vecX_n}^{[0,1]}(\delta) >\beta \big) <\epsilon.
\end{equation}
\end{lem}

\begin{proof}
Notice that
\begin{equation}
\Prob\big( \omega_{\vecX_n}^{[0,1]}(\delta)  >\beta \big)
\leq \Prob\big( \omega_{\xibar^{1/2} \,\vecY_n}^{[0,2\sigmabar]}\circ\omega_{\Theta_n}^{[0,1]}(\delta)  >\beta \big) + \Prob\big( \Theta_n(1+\delta)>2\sigmabar \big) .
\end{equation}
For $n\to\infty$, we have  $\Theta_n(1+\delta) \toprob (1+\delta)\sigmabar <2\sigmabar$ (see the remark after Lemma \ref{pluslem:Zn}), and therefore 
\begin{equation}
\lim_{n\to\infty}\Prob\big( \Theta_n(1+\delta)>2\sigmabar \big) = 0.
\end{equation}
The claim now follows from the tightness of $(\vecY_n)_{n=1}^\infty$ established in Theorem \ref{plusthm:main2}, and from the tightness of $(\Theta_n)_{n=1}^\infty$ which follows from the remark after Lemma \ref{pluslem:Zn}.
\end{proof}

This concludes the proof of Theorem \ref{plusthm:main1}.

\bibliographystyle{amsplain}
%\bibliography{../../jmbibliography}

\providecommand{\bysame}{\leavevmode\hbox to3em{\hrulefill}\thinspace}
\providecommand{\MR}{\relax\ifhmode\unskip\space\fi MR }
% \MRhref is called by the amsart/book/proc definition of \MR.
\providecommand{\MRhref}[2]{%
  \href{http://www.ams.org/mathscinet-getitem?mr=#1}{#2}
}
\providecommand{\href}[2]{#2}

\end{document}